\documentclass[review]{elsarticle}

\usepackage{lineno,hyperref}
\usepackage{amsmath}
\usepackage{amssymb}
\usepackage{algorithmic}
\usepackage{algorithm}
\usepackage{caption}
\usepackage{mathtools}
\usepackage{changes}
\usepackage{multirow}
\usepackage{verbatim}
\usepackage{mathrsfs}
\usepackage{graphicx}
\usepackage{amsmath,amsthm,bm,mathrsfs}

\theoremstyle{plain}% Theorem-like structures provided by amsthm.sty
\newtheorem{theorem}{Theorem}[section]
\newtheorem{lemma}[theorem]{Lemma}

\newtheorem{proposition}[theorem]{Proposition}

\theoremstyle{definition}
\newtheorem{definition}[theorem]{Definition}

\theoremstyle{remark}

\modulolinenumbers[5]

\journal{Journal of \LaTeX\ Templates}

\begin{document}

\begin{frontmatter}

\title{Relative Error-based Time-limited $\mathcal{H}_2$ Model Order Reduction via Oblique Projection}

\author[mymainaddress]{Umair~Zulfiqar}

\author[mymainaddress,mysecondaryaddress]{Xin~Du\corref{mycorrespondingauthor}}
\cortext[mycorrespondingauthor]{Corresponding author}
\ead{duxin@shu.edu.cn}

\author[mymainaddress]{Qiuyan~Song}

\author[ml]{Zhi-Hua~Xiao}

\author[vs]{Victor~Sreeram}

\address[mymainaddress]{School of Mechatronic Engineering and Automation, Shanghai University, Shanghai, 200444, China}
\address[mysecondaryaddress]{Shanghai Key Laboratory of Power Station Automation Technology, Shanghai University, Shanghai, 200444, China}
\address[ml]{School of Information and Mathematics, Yangtze University, Jingzhou, Hubei, 434023, China}
\address[vs]{Department of Electrical, Electronic, and Computer Engineering, The University of Western Australia, Perth, 6009, Australia}

\begin{abstract}
In time-limited model order reduction, a reduced-order approximation of the original high-order model is obtained that accurately approximates the original model within the desired limited time interval. Accuracy outside that time interval is not that important. The error incurred when a reduced-order model is used as a surrogate for the original model can be quantified in absolute or relative terms to access the performance of the model reduction algorithm. The relative error is generally more meaningful than an absolute error because if the original and reduced systems' responses are of small magnitude, the absolute error is small in magnitude as well. However, this does not necessarily mean that the reduced model is accurate. The relative error in such scenarios is useful and meaningful as it quantifies percentage error irrespective of the magnitude of the system's response. In this paper, the necessary conditions for a local optimum of the time-limited $\mathcal{H}_2$ norm of the relative error system are derived. Inspired by these conditions, an oblique projection algorithm is proposed that ensures small $\mathcal{H}_2$-norm relative error within the desired time interval. Unlike the existing relative error-based model reduction algorithms, the proposed algorithm does not require solutions of large-scale Lyapunov and Riccati equations. The proposed algorithm is compared with time-limited balanced truncation, time-limited balanced stochastic truncation, and time-limited iterative Rational Krylov algorithm. Numerical results confirm the superiority of the proposed algorithm over these existing algorithms.
\end{abstract}


\end{frontmatter}

\linenumbers
\section{Introduction}
Computer-aided simulation and analysis have been a major driver in scientific innovation and progress for the last several decades. The dynamical systems and processes, which can be natural or artificial, are mathematically modeled to be used in a computer program that can simulate and analyze that system for endless possible inputs, settings, and scenarios. During the mathematical modeling process, the behavior of the system or process is generally described by partial differential equations (PDEs). To utilize the mathematical tools available in dynamical system theory, the PDEs are converted into ordinary differential equations (ODEs), which are then expressed as a state-space model \cite{schilders2008model,benner2011model}.

Let $u(t)\in\mathbb{R}^{m\times 1}$ be the inputs, $y(t)\in\mathbb{R}^{p\times 1}$ be the outputs, and $x(t)\in\mathbb{R}^{n\times 1}$ be the states of the state-space model of the dynamical system $H$. Then the state-space equations of $H$ can be written as
\begin{align}
\dot{x}(t)&=Ax(t)+Bu(t),\nonumber\\
y(t)&=Cx(t)+Du(t),\nonumber
\end{align} wherein $A\in\mathbb{R}^{n\times n}$, $B\in\mathbb{R}^{n\times m}$, $C\in\mathbb{R}^{p\times n}$, and $D\in\mathbb{R}^{p\times m}$. The state-space realization $(A,B,C,D)$ has the following equivalence with the transfer function of $H$
\begin{align}
H(s)=C(sI-A)^{-1}B+D.\nonumber
\end{align}

Note that the simulation of $H$ requires the solution of $n$ differential equations. The computing power of modern-day computers has been exponentially increasing following Moore's law, and memory resources have been increasing tremendously \cite{benner2011model}. However, the complexity of modern-day systems and processes has also been increasing tremendously following the scientific and technological advancements enabled by high-speed chip manufacturing. Therefore, the value of $n$ in most modern-day systems and processes is normally several thousand. The simulation and analysis of such a high-order model remain a computational challenge, which has been a consistent motivation for developing efficient model order reduction (MOR) algorithms \cite{benner2005dimension,obinata2012model,benner2020model}. MOR is a process of obtaining a reduced-order model (ROM) that has similar properties and behavior but can be simulated cheaply without exhausting the available computational resources. The ROM can then be used as a surrogate for the original model with tolerable numerical error.

Let us denote the ROM as $\hat{H}$. Further, let us denote the states and outputs of $\hat{H}$ as $\hat{x}(t)$ and $\hat{y}(t)$, respectively. Then the state-space equations of $\hat{H}$ can be written as
\begin{align}
\dot{\hat{x}}(t)&=\hat{A}\hat{x}(t)+\hat{B}u(t),\nonumber\\
\hat{y}(t)&=\hat{C}\hat{x}(t)+\hat{D}u(t),\nonumber
\end{align} wherein $\hat{A}\in\mathbb{R}^{r\times r}$, $\hat{B}\in\mathbb{R}^{r\times m}$, $\hat{C}\in\mathbb{R}^{p\times r}$, $\hat{D}\in\mathbb{R}^{p\times m}$, and $r\ll n$. The state-space realization $(\hat{A},\hat{B},\hat{C},\hat{D})$ has the following equivalence with the transfer function of $\hat{H}$
\begin{align}
\hat{H}(s)=\hat{C}(sI-\hat{A})^{-1}\hat{B}+\hat{D}.\nonumber
\end{align}
Ideally, the MOR algorithm should have three main properties:
\begin{enumerate}
  \item $\hat{H}$ should closely mimic $H$ in the frequency and/or time domains.
  \item $\hat{H}$ should preserve important system properties like stability, passivity, contractivity, etc.
  \item The computational cost required to obtain $\hat{H}$ should be way less than that required to simulate $H$.
\end{enumerate}
No MOR algorithm has all these properties; nevertheless, the significance of a MOR algorithm is generally accessed against these three standards. Mathematically, two possible ways to express the MOR problem are the following:
\begin{enumerate}
\item $H=\hat{H}+\Delta_{add}$.
\item $H=\hat{H}(I+\Delta_{rel})$.
\end{enumerate}
Most of the MOR algorithms use additive error $\Delta_{add}=H-\hat{H}$ as their performance criterion, i.e., they minimize $\Delta_{add}$. The quality of approximation in these algorithms is accessed by computing system norms (like $\mathcal{H}_2$ norm and $\mathcal{H}_\infty$ norm) of $\Delta_{add}$. If the output responses of $H$ and $\hat{H}$ are inherently small, a small output response of $\Delta_{add}$ can be a misleading notion of accuracy. The relative error $\Delta_{rel}$ is a better approximation criterion as it reveals percentage error. However, $\Delta_{rel}$ is a less used approximation criterion because most of the algorithms that use this criterion are computationally expensive due to the inherent nature of the problem, which will be discussed later. Various performance specifications and specific properties to be preserved lead to various MOR families. Within the family, several algorithms with different discourses are available, each having some distinct features. In the sequel, a brief literature survey of the important MOR algorithms is presented.

Balanced truncation (BT) is one of the most famous MOR techniques due to its useful theoretical properties, like stability preservation and supreme accuracy \cite{moore1981principal}. BT has led to several important MOR algorithms constituting a rich family of algorithms known as balancing- or gramian-based MOR techniques in the literature. The ROM produced by BT is often suboptimal in the $\mathcal{H}_\infty$ norm, i.e., a suboptimum of $||\Delta_{add}(s)||_{\mathcal{H}_\infty}$. An \textit{apriori} upper bound on $||\Delta_{add}(s)||_{\mathcal{H}_\infty}$ for the ROM generated by BT was proposed in \cite{enns1984model}. Most of the algorithms in the BT family offer a similar \textit{apriori} upper bound, which is a distinct and important feature of balancing-based MOR algorithms. BT is a computationally expensive algorithm as it requires the solution of two large-scale Lyapunov equations, which limits its applicability to medium-scale systems. To overcome this problem, several numerical algorithms are reported that compute low-rank solutions of these Lyapunov equations, which are cheaper to compute \cite{gugercin2005smith,li2002low,penzl1999cyclic,kurschner2020inexact}. By using these low-rank solutions, the applicability of BT is extended to large-scale systems \cite{mehrmann2005balanced}. Some other numerical methods to avoid ill-conditioning in the BT algorithm are also reported, like \cite{tombs1987truncated,safonov1989schur}. BT has been extended to preserve other system properties like passivity and contractivity in \cite{phillips2003guaranteed,yan2008second,phillips2004poor} and \cite{reis2010positive}, respectively. Further, its applicability has also been extended to more general classes of linear and nonlinear systems, cf.\cite{reis2008balanced,wang2011balanced,zhang2003gramians,lall2002subspace}. A closely related yet different MOR method known as optimal Hankel norm approximation (OHNA) is proposed in \cite{glover1984all}. The ROM produced by OHNA is optimal in the Hankel norm, and an \textit{apriori} upper bound on $||\Delta_{add}(s)||_{\mathcal{H}_\infty}$, which is half of that in the case of BT, also holds \cite{glover1989tutorial}. An important extension to BT, which is relevant to the problem considered in this paper, is the balanced stochastic truncation (BST) \cite{green1988balanced}. It minimizes $||\Delta_{rel}(s)||_{\mathcal{H}_{\infty}}$, and an \textit{apriori} upper bound on $||\Delta_{rel}(s)||_{\mathcal{H}_{\infty}}$ also exists \cite{green1988relative}. If the original model is a stable minimum phase system, BST guarantees that the ROM preserves this property.

In a large-scale setting, the computation of the $\mathcal{H}_\infty$ norm is expensive. Moreover, the MOR algorithms that tend to reduce the $\mathcal{H}_{\infty}$ norm of error are generally also expensive \cite{flagg2010interpolation,flagg2013interpolatory,castagnotto2017interpolatory}. The $\mathcal{H}_2$ norm, on the other hand, can be computed cheaply due to its relation with the system gramians, for which several efficient low-rank algorithms are available. Interestingly, the MOR algorithms that seek to ensure less error in the $\mathcal{H}_2$ norm happen to be the most computationally efficient ones in the available literature \cite{wolf2014h}. The selection of the $\mathcal{H}_2$ norm as a performance specification in this paper is motivated by this factor. A locally optimal solution in the $\mathcal{H}_2$ norm is also possible, and several efficient algorithms are available that achieve the local optimum upon convergence. The necessary conditions for a local optimum of $||\Delta_{add}(s)||_{\mathcal{H}_2}^2$ (known as Wilson's conditions) are related to interpolation conditions, i.e., $\hat{H}(s)$ interpolates $H(s)$ at selected points in the $s$-plane, cf. \cite{wilson1970optimum,gugercin2008h_2,van2008h2,xu2011optimal}. The iterative rational Krylov algorithm (IRKA) is the pioneer and most famous interpolation algorithm for single-input single-output (SISO) systems that can achieve the local optimum of $||\Delta_{add}(s)||_{\mathcal{H}_2}^2$ upon convergence \cite{gugercin2008h_2}. IRKA is generalized for multi-input multi-output (MIMO) systems in \cite{van2008h2}. A more general algorithm that does not assume that $H(s)$ and $\hat{H}(s)$ have simple poles is presented in \cite{xu2011optimal} that can capture the local optimum of $||\Delta_{add}(s)||_{\mathcal{H}_2}^2$ upon convergence. This algorithm is based on the Sylvester equations and uses the connection between the Sylvester equations and projection established in \cite{gallivan2004sylvester}. Its numerical properties are further improved in \cite{benner2011sparse}, and a more robust algorithm is presented. To speed up convergence in IRKA, some trust region-based techniques are proposed in \cite{beattie2009trust}. Some other $\mathcal{H}_2$-optimal algorithms based on manifold theory are also reported, like \cite{jiang2017h2,sato2017structure,sato2015riemannian,yang2019trust}. To guarantee stability, some suboptimal solutions to $||\Delta_{add}(s)||_{\mathcal{H}_2}^2$ are reported, like \cite{panzer2013greedy,wolf2013h,ibrir2018projection}.

The simulation of a model does not encompass the entire time horizon since no practical system is run over the entire time range. Thus it is reasonable to ensure good accuracy in the actual time range of operation, which is often limited to a small interval. For instance, low-frequency oscillations in interconnected power systems only last for 15 seconds, after which these are successfully damped by the power system stabilizers and damping controllers. The first 15 seconds thus become important in small-signal stability analysis of interconnected power systems \cite{kundur1994power,pal2006robust,rogers2012power}. Similarly, in the time-limited optimal control problem, the behavior of the plant in the desired time interval is important \cite{grimble1979solution}. This motivates time-limited MOR, wherein it is focused to ensure supreme accuracy in the desired time interval rather than trying to achieve good accuracy over the entire time range. BT was generalized for the time-limited MOR problem in \cite{gawronski1990model} to obtain a time-limited BT (TLBT) algorithm. TLBT does not retain the stability preservation and \textit{apriori} error bound properties of BT. However, an $\mathcal{H}_2$ norm-based upper bound on $\Delta_{add}$ is proposed in \cite{redmann2018output,redmann2020lt2} for TLBT. Some ad hoc approaches to guarantee stability in TLBT are also reported, like \cite{gugercin2004survey}; however, these are quite inaccurate and computationally expensive. TLBT, like BT, is computationally expensive in a large-scale setting and requires solutions of two large-scale Lyapunov equations. In \cite{kurschner2018balanced}, these Lyapunov equations are replaced with their low-rank solutions to extend the applicability of TLBT to large-scale systems. TLBT is further generalized to preserve more properties like passivity and to extend its applicability to a more general class of linear and nonlinear systems in the literature, cf. \cite{zulfiqar2020ptime,benner2021frequency,kumar2017generalized,imran2022development,shaker2014time}. BST is extended to the time-limited MOR scenario in \cite{tahavori2013model}; however, an upper bound on $\Delta_{rel}$ does not exist for time-limited BST (TLBST).

In \cite{goyal2019time}, a definition for the time-limited $\mathcal{H}_{2}$ norm is presented, and the necessary conditions for a local optimum of $||\Delta_{add}(s)||_{\mathcal{H}_{2,\tau}}^2$ (wherein $\mathcal{H}_{2,\tau}$ represents the time-limited $\mathcal{H}_2$ norm) are derived. Then, IRKA is heuristically extended to time-limited IRKA (TLIRKA) for achieving a local optimum of $||\Delta_{add}(s)||_{\mathcal{H}_{2,\tau}}^2$. However, TLIRKA does not satisfy the necessary conditions for a local optimum. In \cite{sinani2019h2}, interpolation-based necessary conditions for a local optimum of $||\Delta_{add}(s)||_{\mathcal{H}_{2,\tau}}^2$ are derived, and a nonlinear optimization-based algorithm is proposed to construct the local optimum. This algorithm, however, is computationally expensive, and its applicability is limited to the order of a few hundred. Another nonlinear optimization-based algorithm for this problem is reported in \cite{das2022h}, which is very similar to that given in \cite{sinani2019h2}. The equivalence between gramian-based conditions and interpolation conditions is established in \cite{zulfiqar2021frequency,zulfiqar2021h2}. A stability-preserving interpolation algorithm that achieves a subset of the necessary conditions is proposed in \cite{zulfiqar2020time}. To the best of the authors' knowledge, there is no algorithm in the literature that seeks to reduce $||\Delta_{rel}(s)||_{\mathcal{H}_{2,\tau}}$, which has motivated the results obtained in this paper.

In this paper, necessary conditions for the local optimum of $||\Delta_{rel}(s)||_{\mathcal{H}_{2,\tau}}^2$ are derived. It is shown that it is inherently not possible to achieve a local optimum of $||\Delta_{rel}(s)||_{\mathcal{H}_{2,\tau}}^2$ in an oblique projection framework. Nevertheless, inspired by these necessary conditions, a computationally efficient oblique projection-based algorithm is proposed that seeks to obtain the local optimum of $||\Delta_{rel}(s)||_{\mathcal{H}_{2,\tau}}^2$ and offers high fidelity. Unlike TLBST, the proposed algorithm does not require the solutions of large-scale Lyapunov and Ricatti equations. The proposed algorithm is numerically compared to TLBT, TLBST, and TLIRKA by considering benchmark examples. The numerical results confirm that the proposed algorithm offers a small relative error within the desired time interval.
\section{Preliminaries}
Let the original high-order system be an $n$-th order $m\times m$ square linear time-invariant system $H$ and $(A,B,C,D)$ be its state-space realization. Further, let us denote the impulse response of $H$ as $h(t)=Ce^{At}B$. The time-limited controllability gramian $P=\int_{0}^{t_d}e^{A\tau}BB^Te^{A^T\tau}d\tau$ and time-limited observability gramian $Q=\int_{0}^{t_d}e^{A^T\tau}C^TCe^{A\tau}d\tau$ of the state-space realization $(A,B,C)$ within the desired time interval $[0,t_d]$ sec solve the following Lyapunov equations
\begin{align}
AP+PA^T+BB^T-e^{At_d}BB^Te^{A^Tt_d}&=0,\nonumber\\
A^TQ+QA+C^TC-e^{A^Tt_d}C^TCe^{At_d}&=0,\nonumber
\end{align}cf. \cite{gawronski1990model}.
\begin{definition}
The time-limited $\mathcal{H}_2$ norm of $H(s)$ is the root mean square of its impulse response within the desired time interval $[0,t_d]$ sec \cite{goyal2019time,sinani2019h2}, i.e.,
\begin{align}
||H(s)||_{\mathcal{H}_{2,\tau}}&=\sqrt{\int_{0}^{t_d}trace\Big(h(\tau)h^T(\tau)\Big)d\tau},\nonumber\\
&=\sqrt{trace(CPC^T)}\nonumber\\
&=\sqrt{\int_{0}^{t_d}trace\Big(h^T(\tau)h(\tau)\Big)d\tau},\nonumber\\
&=\sqrt{trace(B^TQB)}.\nonumber
\end{align}
\end{definition}
\subsection{Problem Formulation}
In an oblique projection framework, the state-space matrices of the ROM are obtained as
\begin{align}
\hat{A}&=\hat{W}^TA\hat{V},&\hat{B}&=\hat{W}^TB,&\hat{C}&=C\hat{V},&\hat{D}&=D,\nonumber
\end{align} wherein $\hat{V}\in\mathbb{R}^{n\times r}$, $\hat{W}\in\mathbb{R}^{n\times r}$, $\hat{W}^T\hat{V}=I$, the columns of $\hat{V}$ span an $r$-dimensional subspace along with the kernels of $\hat{W}^T$, and $\Pi=\hat{V}\hat{W}^T$ is an oblique projector \cite{antoulas2005approximation}.

In time-limited MOR, it is aimed to ensure that the difference $y(t)-\hat{y}(t)$ is small within the desired time interval $[0,t_d]$ sec. The $\mathcal{H}_{2,\tau}$-norm is an effective measure to quantify the error in the desired time interval \cite{goyal2019time,sinani2019h2}. The problem of relative error-based time-limited $\mathcal{H}_2$ MOR via oblique projection is to construct the reduction matrices $\hat{V}\in\mathbb{R}^{n\times r}$ and $\hat{W}\in\mathbb{R}^{n\times r}$ such that the ROM $\hat{H}(s)$ obtained with these matrices ensures a small $\mathcal{H}_{2,\tau}$-norm of the relative error $\Delta_{rel}(s)$, i.e.,
\begin{align}
\underset{\substack{\hat{H}(s)\\\textnormal{order}=r}}{\text{min}}||\Delta_{rel}(s)||_{\mathcal{H}_{2,\tau}}.\nonumber
\end{align}
\subsection{Existing Oblique Projection Techniques}
In this subsection, three important oblique projection algorithms for time-limited MOR are briefly reviewed. These algorithms are the most relevant to the new results obtained in this paper.
\subsubsection{Time-limited Balanced Truncation (TLBT)}
In TLBT \cite{gawronski1990model}, the reduction matrices are computed from the contragradient transformation of $P$ and $Q$ as $\hat{W}^TP\hat{W}=\hat{V}^TQ\hat{V}= diag(\sigma_1,\cdots,\sigma_r)$ where $\sigma_{i+1}\geq\sigma_i$. The resultant ROM offers a small additive error $\Delta_{add}$ in the desired time interval $[0,t_d]$ sec.
\subsubsection{Time-limited Balanced Stochastic Truncation (TLBST)}
The controllability gramian $W_c$ of the pair $(A,B)$ solves the following Lyapunov equation
\begin{align}
AW_c+W_cA^T+BB^T=0.\nonumber
\end{align} Define $B_s$ and $A_s$ as the following
\begin{align}
B_s&=W_cC^T+BD^T,&A_s&=A-B_s(DD^T)^{-1}C.\nonumber
\end{align}
Now, let $X_s$ solves the following Riccati equation
\begin{align}
A_s^TX_s+X_s A_s+X_s B_s(DD^T)^{-1}B_s^TX_s+C^T(DD^T)^{-1}C=0.\nonumber
\end{align}
The time-limited observability gramian $X_\tau$ of the pair $\big(A,D^{-1}(C-B_s^TX_s)\big)$ solves the following Lyapunov equation
\begin{align}
A^TX_\tau+X_\tau A&+\big(D^{-1}(C-Bs^TX_s)\big)^T\big(D^{-1}(C-Bs^TX_s)\big)\nonumber\\
&-e^{A^Tt_d}\big(D^{-1}(C-Bs^TX_s)\big)^T\big(D^{-1}(C-Bs^TX_s)\big)e^{At_d}=0.\nonumber
\end{align}
In TLBST \cite{tahavori2013model}, the reduction matrices are computed from the contragradient transformation of $P$ and $X_\tau$ as $\hat{W}^TP\hat{W}=\hat{V}^TX_\tau\hat{V}= diag(\sigma_1,\cdots,\sigma_r)$ where $\sigma_{i+1}\geq\sigma_i$. The resultant ROM offers a small relative error $\Delta_{rel}$ in the desired time interval $[0,t_d]$ sec.
\subsubsection{Time-limited Iterative Rational Krylov Algorithm (TLIRKA)}
TLIRKA starts with an arbitrary guess of the ROM's state-space matrices $(\hat{A},\hat{B},\hat{C})$, wherein $\hat{A}$ has simple eigenvalues \cite{goyal2019time}. Let $\hat{A}=\hat{S}\Lambda\hat{S}^{-1}$ be the eigenvalue decomposition of $\hat{A}$. Further, let $P_{12}$ and $Y_\tau$ satisfy the following Sylvester equations
\begin{align}
AP_{12}+P_{12}\hat{A}^T+B\hat{B}^T-e^{At_d}B\hat{B}^Te^{\hat{A}^Tt_d}&=0,\label{e6}\\
A^TY_\tau+Y_\tau\hat{A}-C^T\hat{C}+e^{A^Tt_d}C^T\hat{C}e^{\hat{A}t_d}&=0.
\end{align} The reduction matrices are computed as $\hat{V}=orth(P_{12}\hat{S}^{-*})$ and $\hat{W}=orth\big(Y_\tau\hat{S}$ $(\hat{V}^*Y_\tau\hat{S})^{-1}\big)$, wherein $orth(\cdot)$ and $(\cdot)^*$ represent orthogonal basis and Hermitian, respectively. The ROM is updated using these reduction matrices, and the process is repeated until the relative change in the eigenvalues of $\hat{A}$ stagnates. Upon convergence, a ROM that offers a small $||\Delta_{add}(s)||_{\mathcal{H}_{2,\tau}}$ is achieved.
\section{Main Results}
In this section, first, a state-space realization and an expression for the $\mathcal{H}_{2,\tau}$ norm of $\Delta_{rel}(s)$ are obtained. Second, necessary conditions for the local optimum of $||\Delta_{rel}(s)||_{\mathcal{H}_{2,\tau}}^2$ are derived. Third, an oblique projection algorithm is presented, which seeks to satisfy these conditions. Fourth, the algorithmic and computational aspects of the proposed algorithm are discussed.
\subsection{State-space Realization and $\mathcal{H}_{2,\tau}$ norm of $\Delta_{rel}(s)$}
Let us assume that $\hat{H}(s)$ is invertible and stable, i.e., $\hat{H}(s)$ is a stable minimum phase transfer function. Then, it can readily be noted that $\Delta_{rel}(s)$ can be expressed as $\Delta_{rel}(s)=\hat{H}^{-1}(s)\Delta_{add}(s)$. Further, if the matrix $D$ is full rank, the following state-space realization exists for $\hat{H}^{-1}(s)$
\begin{align}
A_i&=\hat{A}-\hat{B}D^{-1}\hat{C},& B_i&=-\hat{B}D^{-1},& C_i&=D^{-1}\hat{C},& D_i&=D^{-1},\nonumber
\end{align}cf. \cite{zhou1995frequency,zhou1996robust}. Then, a possible state-space realization of $\hat{H}^{-1}(s)\Delta_{add}(s)$ is given as
\begin{align}
A_{rel}&=\begin{bmatrix}A&0&0\\0&\hat{A}&0\\B_iC& -B_i\hat{C}&A_i\end{bmatrix},&B_{rel}&=\begin{bmatrix}B\\\hat{B}\\0\end{bmatrix},\nonumber\\
C_{rel}&=\begin{bmatrix}D_iC&-D_i\hat{C}&C_i\end{bmatrix}.\nonumber
\end{align}
Due to the triangular structure of $A_{rel}$, the matrix exponential $e^{A_{rel}t_d}$ can be expressed as
\begin{align}
e^{A_{rel}t_d}=\begin{bmatrix}e^{At_d}&0&0\\0& e^{\hat{A}t_d}&0\\E_1 &E_2 &e^{A_it_d}\end{bmatrix},\nonumber
\end{align}cf. \cite{bini2016computing}. Since $A_{rel}e^{A_{rel}t_d}=e^{A_{rel}t_d}A_{rel}$, cf. \cite{smalls2007exponential}, the following holds
\begin{align}
\begin{bsmallmatrix}\star&&&&\star&&\star\\\star&&&&\star&&\star\\A_iE_1-E_1A+B_iCe^{At_d}-e^{A_it_d}B_iC&&&&A_iE_2-E_2\hat{A}+B_i\hat{C}e^{\hat{A}t_d}-e^{A_it_d}B_i\hat{C}&&\star\end{bsmallmatrix}=0.\nonumber
\end{align} Thus $E_1$ and $E_2$ solve the following Sylvester equations
\begin{align}
A_iE_1-E_1A+B_iCe^{At_d}-e^{A_it_d}B_iC&=0,\label{e8}\\
A_iE_2-E_2\hat{A}+B_i\hat{C}e^{\hat{A}t_d}-e^{A_it_d}B_i\hat{C}&=0.\label{e9}
\end{align}
\begin{proposition}
$E_2=e^{\hat{A}t_d}-e^{A_it_d}$.
\end{proposition}
\begin{proof}
\begin{align}
A_iE_2-E_2\hat{A}&=B_i\hat{C}e^{\hat{A}t_d}-e^{A_it_d}B_i\hat{C}\nonumber\\
A_iE_2-E_2\hat{A}&=B_i\hat{C}e^{\hat{A}t_d}-e^{A_it_d}B_i\hat{C}+B_i\hat{C}e^{A_it_d}-B_i\hat{C}e^{A_it_d}\nonumber
\end{align}
Since $A_ie^{A_it_d}-e^{A_it_d}A_i=0$, it can readily be noted that
\begin{align}
B_i\hat{C}e^{A_it_d}-e^{A_it_d}B_i\hat{C}=e^{A_it_d}\hat{A}-\hat{A}e^{A_it_d}.\nonumber
\end{align}
Thus
\begin{align}
A_iE_2-E_2\hat{A}&=B_i\hat{C}e^{\hat{A}t_d}-e^{A_it_d}B_i\hat{C}+e^{A_it_d}\hat{A}-\hat{A}e^{A_it_d}\nonumber\\
&=B_i\hat{C}e^{\hat{A}t_d}-e^{A_it_d}B_i\hat{C}+e^{A_it_d}\hat{A}-\hat{A}e^{A_it_d}+\hat{A}e^{\hat{A}t_d}-e^{\hat{A}t_d}\hat{A}\nonumber\\
&=A_i(e^{\hat{A}t_d}-e^{A_it_d})-(e^{\hat{A}t_d}-e^{A_it_d})\hat{A}.\nonumber
\end{align}
Due to the uniqueness of the Sylvester equation (\ref{e9}), $E_2=e^{\hat{A}t_d}-e^{A_it_d}$.
\end{proof}

Let $P_{rel}=\begin{bmatrix}P&P_{12}&P_{13}\\P_{12}^T&P_{22}&P_{23}\\P_{13}^T&P_{23}^T&P_{33}\end{bmatrix}$ be the time-limited controllability gramian of the pair $(A_{rel},B_{rel})$, which solves the following Lyapunov equation
\begin{align}
A_{rel}P_{rel}+P_{rel}A_{rel}^T+B_{rel}B_{rel}^T-e^{A_{rel}t_d}B_{rel}B_{rel}^Te^{A_{rel}^Tt_d}&=0.\label{e10}
\end{align} By expanding the equation (\ref{e10}), it can be noted that $P_{13}$, $P_{22}$, $P_{23}$, and $P_{33}$ solve the following linear matrix equations
\begin{align}
AP_{13}+P_{13}A_i^T+K_{13}&=0,\label{e11}\\
\hat{A}P_{22}+P_{22}\hat{A}^T+K_{22}&=0,\label{e12}\\
\hat{A}P_{23}+P_{23}A_i^T+K_{23}&=0,\label{e13}\\
A_iP_{33}+P_{33}A_i^T+K_{33}&=0,\label{e14}
\end{align}wherein
\begin{align}
K_{13}&=P_{11}C^TB_i^T-P_{12}\hat{C}^TB_i^T-e^{At_d}BB^TE_1^T-e^{At_d}B\hat{B}^TE_2^T,\nonumber\\
K_{22}&=\hat{B}\hat{B}^T-e^{\hat{A}t_d}\hat{B}\hat{B}^Te^{\hat{A}^Tt_d},\nonumber\\
K_{23}&=P_{12}^TC^TB_i^T-P_{22}\hat{C}^TB_i^T-e^{\hat{A}t_d}\hat{B}B^TE_1^T-e^{\hat{A}t_d}\hat{B}\hat{B}^TE_2^T,\nonumber\\
K_{33}&=B_iCP_{13}-B_i\hat{C}P_{23}+P_{13}^TC^TB_i^T-P_{23}^T\hat{C}^TB_i^T\nonumber\\
&\hspace{1.65cm}-E_1BB^TE_1^T-E_2\hat{B}B^TE_1^T-E_1B\hat{B}^TE_2^T-E_2\hat{B}\hat{B}^TE_2^T.\nonumber
\end{align}
Let $Q_{rel}=\begin{bmatrix}Q_{11}&Q_{12}&Q_{13}\\Q_{12}^T&Q_{22}&Q_{23}\\Q_{13}^T&Q_{23}&Q_{33}\end{bmatrix}$ be the time-limited observability gramian of the pair $(A_{rel},C_{rel})$, which solves the following Lyapunov equation
\begin{align}
A_{rel}^TQ_{rel}+Q_{rel}A_{rel}+C_{rel}^TC_{rel}-e^{A_{rel}^Tt_d}C_{rel}^TC_{rel}e^{A_{rel}t_d}&=0.\label{e15}
\end{align}
By expanding the equation (\ref{e15}), it can be noted that $Q_{11}$, $Q_{12}$, $Q_{13}$, $Q_{22}$, $Q_{23}$, and $Q_{33}$ solve the following linear matrix equations
\begin{align}
A^TQ_{11}+Q_{11}A+L_{11}&=0,\label{e16}\\
A^TQ_{12}+Q_{12}\hat{A}+L_{12}&=0,\label{e17}\\
A^TQ_{13}+Q_{13}A_i+L_{13}&=0,\label{e18}\\
\hat{A}^TQ_{22}+Q_{22}\hat{A}+L_{22}&=0,\label{e19}\\
\hat{A}^TQ_{23}+Q_{23}A_i+L_{23}&=0,\label{e20}\\
A_i^TQ_{33}+Q_{33}A_i+L_{33}&=0,\label{e21}
\end{align}wherein
\begin{align}
L_{11}&=C^TB_i^TQ_{13}^T+Q_{13}B_iC+C^TD_i^TD_iC-e^{A^Tt_d}C^TD_i^TD_iCe^{At_d}\nonumber\\
&\hspace*{2.25cm}-E_1^TC_i^TD_iCe^{At_d}-e^{A^Tt_d}C^TD_i^TC_iE_1-E_1^TC_i^TC_iE_1,\nonumber\\
L_{12}&=C^TB_i^TQ_{23}-Q_{13}B_i\hat{C}-C^TD_i^TD_i\hat{C}+e^{A^Tt_d}C^TD_i^TD_i\hat{C}e^{\hat{A}t_d}\nonumber\\
&\hspace*{2.25cm}+E_1^TC_i^TD_i\hat{C}e^{\hat{A}t_d}-e^{A^Tt_d}C^TD_i^TC_iE_2-E_1^TC_i^TC_iE_2,\nonumber\\
L_{13}&=C^TB_i^TQ_{33}+C^TD_i^TC_i-e^{A^Tt_d}C^TD_i^TC_ie^{A_it_d}-E_1^TC_i^TC_ie^{A_it_d},\nonumber\\
L_{22}&=-\hat{C}^TB_i^TQ_{23}-Q_{23}B_i\hat{C}+\hat{C}^TD_i^TD_i\hat{C}-e^{\hat{A}^Tt_d}\hat{C}^TD_i^TD_i\hat{C}e^{\hat{A}t_d}\nonumber\\
&\hspace*{2.25cm}+E_2^TC_i^TD_i\hat{C}e^{\hat{A}t_d}+e^{\hat{A}^Tt_d}\hat{C}^TD_i^TC_iE_2-E_2^TC_i^TC_iE_2,\nonumber\\
L_{23}&=-\hat{C}^TB_i^TQ_{33}-\hat{C}^TD_i^TC_i+e^{\hat{A}^Tt_d}\hat{C}^TD_i^TC_ie^{A_it_d}-E_2^TC_i^TC_ie^{A_it_d},\nonumber\\
L_{33}&=C_i^TC_i-e^{A_i^Tt_d}C_i^TC_ie^{A_it_d}.\nonumber
\end{align}
\begin{proposition}
$Q_{12}=-Q_{13}$ and $Q_{22}=-Q_{23}=Q_{33}$.
\end{proposition}
\begin{proof}
By adding the equations (\ref{e20}) and (\ref{e21}), we get
\begin{align}
A_i^T(Q_{23}+Q_{33})+(Q_{23}+Q_{33})A_i&=0.\nonumber
\end{align}
Thus $Q_{23}+Q_{33}=0$ and $Q_{33}=-Q_{23}$. Similarly, by subtracting the equation (\ref{e19}) from (\ref{e21}), we get
\begin{align}
A_i^T(Q_{33}-Q_{22})+(Q_{33}-Q_{22})A_i&=0.\nonumber
\end{align}
Thus $Q_{33}-Q_{22}=0$ and $Q_{33}=Q_{22}$. Finally, by adding the equations (\ref{e16}) and (\ref{e17}), we get
\begin{align}
A^T(Q_{12}+Q_{13})+(Q_{12}+Q_{13})\hat{A}&=0.\nonumber
\end{align}
Thus $Q_{12}+Q_{13}=0$ and $Q_{12}=-Q_{13}$.
\end{proof}
From the definition of the $\mathcal{H}_{2,\tau}$ norm, $||\Delta_{rel}(s)||_{\mathcal{H}_{2,\tau}}$ can be written as
\begin{align}
||\Delta_{rel}(s)||_{\mathcal{H}_{2,\tau}}&=\sqrt{trace(C_{rel}P_{rel}C_{rel}^T)}\nonumber\\
&=\big(trace(D_iCPC^TD_i^T-2D_iCP_{12}\hat{C}^TD_i^T+2D_iCP_{13}C_i^T\nonumber\\
&\hspace*{3cm}+D_i\hat{C}P_{22}\hat{C}^TD_i^T-2D_i\hat{C}P_{23}C_i^T+C_iP_{33}C_i^T)\big)^{\frac{1}{2}}\nonumber\\
&=\sqrt{trace(B_{rel}^TQ_{rel}B_{rel})}\nonumber\\
&=\big(trace(B^TQ_{11}B+2B^TQ_{12}\hat{B}+\hat{B}^TQ_{22}\hat{B})\big)^{\frac{1}{2}}.\nonumber
\end{align}
\subsection{Necessary Conditions for the Local optimum of $||\Delta_{rel}(s)||_{\mathcal{H}_{2,\tau}}^2$}
The derivation of the necessary conditions for a local optimum of $||\Delta_{rel}(s)||_{\mathcal{H}_{2,\tau}}^2$ requires several new variables to be defined. For clarity, we define new variables before presenting the optimality conditions. Let us define $\xi_1$, which solves the following Sylvester equation
\begin{align}
A_i^T\xi_1-\xi_1A^T+O_1&=0\label{e22}
\end{align}where
\begin{align}
O_1=-C_i^TC_iE_1X_{11}-C_i^TC_iE_2X_{12}^T-C_i^TD_iCe^{At_d}X_{11}+C_i^TD_i\hat{C}e^{\hat{A}t_d}X_{12}^T.\nonumber
\end{align}
$\xi_1$ will be used in deriving the partial derivative of $||\Delta_{rel}(s)||_{\mathcal{H}_{2,\tau}}^2$ with respect to $\hat{A}$. Next, let us define $X_{11}$, $X_{12}$, $X_{13}$, $X_{22}$, and $X_{33}$, which solve the following linear matrix equations
\begin{align}
AX_{11}+X_{11}A^T+BB^T&=0,\hspace*{0.5cm}\label{e23}\\
AX_{12}+X_{12}\hat{A}^T+B\hat{B}^T&=0,\label{e24}\\
AX_{13}+X_{13}A_i^T+X_{11}C^TB_i^T-X_{12}\hat{C}^TB_i^T&=0,\label{e25}\\
\hat{A}X_{22}+X_{22}\hat{A}^T+\hat{B}\hat{B}^T&=0,\label{e26}\\
\hat{A}X_{23}+X_{23}A_i^T+X_{12}^TC^TB_i^T-X_{22}\hat{C}^TB_i^T&=0,\label{e27}\\
A_iX_{33}+X_{33}A_i^T+B_iCX_{13}-B_i\hat{C}X_{23}+X_{13}^TC^TB_i^T-X_{23}^T\hat{C}^TB_i^T&=0.\label{e28}
\end{align}
These variables will be used in deriving the partial derivative of $||\Delta_{rel}(s)||_{\mathcal{H}_{2,\tau}}^2$ with respect to $\hat{A}$ and $\hat{B}$.  Next, let us define $\xi_2$, which solves the following Sylvester equation
\begin{align}
A_i^T\xi_2-\xi_2A^T+O_2&=0\label{e29}
\end{align}where
\begin{align}
O_2&=-C_i^TC_ie^{A_it_d}X_{13}^T-C_i^TD_iCe^{At_d}X_{11}-C_i^TC_iE_1X_{11}+C_i^TD_i\hat{C}e^{\hat{A}t_d}X_{12}^T\nonumber\\
&\hspace*{3cm}-2C_i^TC_iE_2X_{12}^T.\nonumber
\end{align}$\xi_2$ will be used in deriving the partial derivative of $||\Delta_{rel}(s)||_{\mathcal{H}_{2,\tau}}^2$ with respect to $\hat{B}$. Next, let us define $Y_{13}$, $Y_{23}$, $Y_{33}$, and $\xi_3$, which solve the following linear matrix equations
\begin{align}
A^TY_{13}+Y_{13}A_i+C^TB_i^TY_{33}+C^TD_i^TC_i&=0,\label{e30}\\
\hat{A}^TY_{23}+Y_{23}A_i-\hat{C}^TB_i^TY_{33}-\hat{C}^TD_i^TC_i&=0,\label{e31}\\
A_i^TY_{33}+Y_{33}A_i+C_i^TC_i&=0,\label{e32}\\
A_i^T\xi_3-\xi_3A^T-O_3&=0,\label{e33}
\end{align}
wherein
\begin{align}
O_3=Y_{13}^Te^{At_d}BB^T-Y_{23}^Te^{\hat{A}t_d}\hat{B}B^T-Y_{33}E_1BB^T-Y_{33}E_2\hat{B}B^T.\nonumber
\end{align}These variables will be used in deriving the partial derivative of $||\Delta_{rel}(s)||_{\mathcal{H}_{2,\tau}}^2$ with respect to $\hat{C}$.

The following properties of trace will also be used repeatedly in the derivation of the optimality conditions:
\begin{enumerate}
  \item Addition: $trace(A_1+A_2+A_3)=trace(A_1)+trace(A_2)+trace(A_3)$.
  \item Transpose: $trace(A_1^T)=trace(A_1)$.
  \item Cyclic permutation: $trace(A_1A_2A_3)=trace(A_3A_1A_2)=trace(A_2A_3A_1)$.
  \item Partial derivative: $\Delta_{f(A)}^A=trace\Big(\frac{\partial}{\partial A}\big(f(A)\big)(\Delta_{A})^T\Big)$ where $\Delta_{f(A)}^A$ is the first-order derivative of $f(A)$ with respect to $A$, and $\Delta_{A}$ is the differential of $A$,
\end{enumerate}cf. \cite{petersen2008matrix}.

Lastly, the following lemma will also be used repeatedly in the derivation.
\begin{lemma}\label{lemma}
Let $U$ and $V$ satisfy the following Sylvester equations
\begin{align}
XU+UY+W&=0,\nonumber\\
YV+VX+Z&=0.\nonumber
\end{align}
Then $trace(WV)=trace(ZU)$.
\end{lemma}
\begin{proof}
See the proof of Lemma 4.1 in \cite{petersson2013nonlinear}.
\end{proof}
We are now in a position to state the necessary conditions, which a local optimum ROM $(\hat{A},\hat{B},\hat{C},D)$ must satisfy, in the following theorem.
\begin{theorem}
Suppose $\hat{H}(s)$ is a stable minimum phase system, and the matrix $D$ is full rank. Then, the local optimum $(\hat{A},\hat{B},\hat{C},D)$ of $||\Delta_{rel}(s)||_{\mathcal{H}_{2,\tau}}^2$ must satisfy the following optimality conditions
\begin{align}
Q_{12}^TX_{12}+Q_{22}X_{22}+\zeta_1&=0,\label{e34}\\
Q_{12}^TB+Q_{22}\hat{B}+\zeta_2&=0,\label{e35}\\
-D_i^TD_iCP_{12}+D_i^TD_i\hat{C}P_{22}+\zeta_3&=0,\label{e36}
\end{align}wherein
\begin{align}
\zeta_1&=Q_{13}^TX_{13}+2Q_{23}X_{23}+Q_{23}X_{23}^T+Q_{33}X_{33}+t_de^{A_i^Tt_d}C_i^TD_iCe^{At_d}X_{12}\nonumber\\
      &+t_de^{A_i^Tt_d}C_i^TC_iE_1X_{12}-t_de^{A_i^Tt_d}C_i^TD_iCe^{At_d}X_{13}-t_de^{A_i^Tt_d}C_i^TC_iE_1X_{13}\nonumber\\
      &-t_de^{A_i^Tt_d}C_i^TC_ie^{A_it_d}X_{22}+t_de^{A_i^Tt_d}C_i^TC_ie^{A_it_d}X_{23}^T+t_de^{A_i^Tt_d}C_i^TC_ie^{A_it_d}X_{23}\nonumber\\
      &-t_de^{A_i^Tt_d}C_i^TC_ie^{A_it_d}X_{33}+\xi_1E_1^T-2t_de^{A_i^Tt_d}\xi_1C^TB_i^T,\nonumber\\
\zeta_2&=-2Q_{13}^TX_{11}C^TD_i^T+Q_{23}X_{22}C_i^T-Q_{23}X_{12}^TC^TD_i^T+Q_{13}^TX_{12}C_i^T\nonumber\\
      &-Q_{33}X_{13}^TC^TD_i^T-Q_{13}^TX_{13}C_i^T-Q_{23}X_{23}C_i^T+Q_{33}X_{23}^TC_i^T-Q_{33}X_{33}C_i^T\nonumber\\
      &-t_de^{A_i^Tt_d}C_i^TD_iCe^{At_d}X_{12}C_i^T-t_de^{A_i^Tt_d}C_i^TC_iE_1X_{12}C_i^T\nonumber\\
      &+t_de^{A_i^Tt_d}C_i^TD_iCe^{At_d}X_{13}C_i^T+t_de^{A_i^Tt_d}C_i^TC_iE_1X_{13}C_i^T\nonumber\\
      &+t_de^{A_i^Tt_d}C_i^TC_ie^{A_it_d}X_{22}C_i^T-t_de^{A_i^Tt_d}C_i^TC_ie^{A_it_d}X_{23}C_i^T\nonumber\\
      &-t_de^{A_i^Tt_d}C_i^TC_ie^{A_it_d}X_{23}^TC_i^T+t_de^{A_i^Tt_d}C_i^TC_ie^{A_it_d}X_{33}C_i^T\nonumber\\
      &-\xi_2E_1^TC_i^T-\xi_2e^{A^Tt_d}C^TD_i^T+e^{A_i^Tt_d}\xi_2C^TD_i^T+t_de^{A_i^Tt_d}\xi_2C^TB_i^TC_i^T,\nonumber\\
\zeta_3&=D_i^TD_iCP_{13}-D_i^TD_i\hat{C}P_{23}-D_i^TC_iP_{23}^T+D_i^TC_iP_{33}+B_i^TY_{13}^TP_{13}\nonumber\\
      &+2B_i^TY_{13}^TP_{12}+B_i^TY_{23}P_{23}-B_i^TY_{23}P_{22}-B_i^TY_{33}P_{23}^T\nonumber\\
      &+t_dB_i^Te^{A_i^Tt_d}Y_{13}^Te^{At_d}B\hat{B}^T+t_dB_i^Te^{A_i^Tt_d}Y_{23}e^{\hat{A}t_d}\hat{B}\hat{B}^T+t_dB_i^Te^{A_i^Tt_d}Y_{33}E_2\hat{B}\hat{B}^T\nonumber\\
      &+t_dB_i^Te^{A_i^Tt_d}Y_{33}E_1B\hat{B}^T-t_dB_i^Te^{A_i^Tt_d}\xi_3\hat{C}^TB_i^T+B_i^T\xi_3E_1^T.\nonumber
\end{align}
\end{theorem}
\begin{proof}
The proof is long and tedious and can be found in the appendix given at the end.
\end{proof}
\subsection{An Oblique Projection Algorithm}
In this subsection, we sketch an oblique projection algorithm from the optimality conditions derived in the last subsection. Some numerical and computational issues related to the applicability of the algorithm are also discussed.
\subsubsection{Limitation in the Oblique Projection Framework}
Let us assume that $P_{22}$ and $Q_{22}$ are invertible. Then, the optimal choices of $\hat{B}$ and $\hat{C}$ are given as
\begin{align}
\hat{B}&=-Q_{22}^{-1}Q_{12}^TB-Q_{22}^{-1}\xi_2,&\hat{C}&=CP_{12}P_{22}^{-1}-(D_i^TD_i)^{-1}\xi_3P_{22}^{-1}.\nonumber
\end{align} If the ROM is obtained via oblique projection, this suggests selecting $\hat{V}$ and $\hat{W}$ as $\hat{V}=P_{12}P_{22}^{-1}$ and $\hat{W}=-Q_{12}Q_{22}^{-1}$, respectively. Due to the oblique projection condition $\hat{W}^T\hat{V}=I$, we get
\begin{align}
Q_{12}^TP_{12}+Q_{22}P_{22}=0.\nonumber
\end{align}
To compute the deviation in the satisfaction of optimality conditions with this choice of reduction subspaces, let us express the optimality condition (\ref{e34}) a bit differently. Note that
\begin{align}
X_{12}&=P_{12}+e^{At_d}X_{12}e^{\hat{A}^Tt_d}& &\textnormal{and}& X_{22}&=P_{22}+e^{\hat{A}t_d}X_{22}e^{\hat{A}^Tt_d},\nonumber
\end{align} cf. \cite{zulfiqar2020time,zulfiqar2021frequency}.
Thus
\begin{align}
Q_{12}^TX_{12}+Q_{22}X_{22}&=Q_{12}^TP_{12}+Q_{22}P_{22}+Q_{12}^Te^{At_d}X_{12}e^{\hat{A}^Tt_d}+Q_{22}e^{\hat{A}t_d}X_{22}e^{\hat{A}^Tt_d}.\nonumber
\end{align}
Then, the optimality condition (\ref{e34}) can be rewritten as
\begin{align}
Q_{12}^TP_{12}+Q_{22}P_{22}+\bar{\xi_1}&=0\nonumber
\end{align}where
\begin{align}
\bar{\xi_1}=\xi_1+Q_{12}^Te^{At_d}X_{12}e^{\hat{A}^Tt_d}+Q_{22}e^{\hat{A}t_d}X_{22}e^{\hat{A}^Tt_d}.\nonumber
\end{align}
It is now clear that by selecting the reduction matrices as $\hat{V}=P_{12}P^{-1}_{22}$ and $\hat{W}=-Q_{12}Q^{-1}_{22}$, deviations in the satisfaction of the optimality conditions (\ref{e34}), (\ref{e35}), and (\ref{e36}) are given by $\bar{\xi_1}$, $\xi_2$, and $\xi_3$, respectively. In general, $\bar{\xi_1}\neq 0$, $\xi_2\neq 0$, and $\xi_3\neq 0$. Thus it is inherently not possible to obtain a local optimum of $||\Delta_{rel}(s)||_{\mathcal{H}_{2,\tau}}^2$ in an oblique projection framework. Nevertheless, the reduction matrices $\hat{V}=P_{12}P^{-1}_{22}$ and $\hat{W}=-Q_{12}Q^{-1}_{22}$ target the optimality conditions and seek to construct a local optimum.
\subsubsection{Algorithm}
The appropriate reduction matrices for the problem under consideration have so far been identified from the optimality conditions. We now proceed in sketching an algorithm with this choice of reduction matrices. The matrices $P_{12}$, $P_{22}$, $Q_{12}$, and $Q_{22}$ depend on the ROM, which is unknown. Thus an iterative solution seems inevitable. Note that $\hat{V}=P_{12}$ and $\hat{W}=Q_{12}$ span the same subspace, as $P^{-1}_{22}$ and $-Q^{-1}_{22}$ only change the basis, cf. \cite{gallivan2004sylvester}. The condition on invertibility of $P_{22}$ and $Q_{22}$ in every iteration can be restrictive and cause numerical ill-conditioning or even failure. Therefore, $\hat{V}=P_{12}$ and $\hat{W}=Q_{12}$ seem a better choice for the reduction matrices. Further, note that $(\hat{A},\hat{B},\hat{C})$ and $(\hat{V},\hat{W})$ can be seen as two coupled systems, i.e.,
\begin{align}
f(\hat{V},\hat{W})&=(\hat{A},\hat{B},\hat{C})&\textnormal{and}&&g(\hat{A},\hat{B},\hat{C})&=(\hat{V},\hat{W}).\nonumber
\end{align}The problem of computing the ROM $(\hat{A},\hat{B},\hat{C})$ can be seen as a stationary point problem, i.e.,
\begin{align}
(\hat{A},\hat{B},\hat{C})=f\big(g(\hat{A},\hat{B},\hat{C})\big)\nonumber
\end{align} with an additional constraint that $\hat{W}^T\hat{V}=I$. Starting with an arbitrary guess of $(\hat{A},\hat{B},\hat{C})$, the process can be continued until convergence to obtain the stationary points. There are still two issues to be addressed before sketching a stationary point solution algorithm. Note that it is so far assumed that the matrix $D$ is full rank, which essentially ensures $D^{-1}$ exists. If the matrix $D$ is not full rank, we can use a remedy employed in BST, i.e., the so-called $\epsilon$-regularization. The rank deficient $D$ can be replaced with $D=\epsilon I$, wherein $\epsilon$ is a small positive number enough to ensure the invertibility of $D$ \cite{benner2001efficient}. In the final ROM, the original $D$ is plugged back in. This is an effective approach and is used in MATLAB's built-in function $``bst()"$ for BST. The second assumption we used is that a stable realization of $\hat{H}^{-1}(s)$ exists because $\hat{H}(s)$ is a stable minimum phase system. Ensuring this condition in every iteration can be a serious limitation in the applicability of the proposed stationary point iteration algorithm. A solution to this problem is discussed in the sequel.

The classical $H_{2}$ norm of $\Delta_{rel}(s)$ can be expressed in the frequency domain as the following
\begin{align}
||\Delta_{rel}(s)||_{\mathcal{H}_2}&=\sqrt{\frac{1}{2\pi}\int_{-\infty}^{\infty}trace\Big(\Delta_{rel}^*(j\omega)\Delta_{rel}(j\omega)\Big)d\omega}\nonumber\\
  &=\sqrt{\frac{1}{2\pi}\int_{-\infty}^{\infty}trace\Big(\Delta_{add}^*(j\omega)\hat{H}^{-*}(j\omega)\hat{H}^{-1}(j\omega)\Delta_{add}(j\omega)\Big)d\omega}\nonumber.
\end{align}
Let $\hat{G}^*(s)$ be a minimum phase right spectral factor of $\hat{H}^*(s)\hat{H}(s)$, such that $\hat{G}(s)\hat{G}^*(s)=\hat{H}^*(s)\hat{H}(s)$. Then $||\Delta_{rel}(s)||_{\mathcal{H}_2}$ can be rewritten as
  \begin{align}
  ||\Delta_{rel}(s)||_{\mathcal{H}_2}
  &=\sqrt{\frac{1}{2\pi}\int_{-\infty}^{\infty}trace\Big(\Delta_{add}^*(j\omega)\hat{G}^{-1}(j\omega)\hat{G}^{-*}(j\omega)\Delta_{add}(j\omega)\Big)d\omega}\nonumber\\
  &=\sqrt{\frac{1}{2\pi}\int_{-\infty}^{\infty}trace\Big(\tilde{\Delta}_{rel}^*(j\omega)\tilde{\Delta}_{rel}(j\omega)\Big)d\omega}.\nonumber
  \end{align}
Let us denote the impulse response of $\tilde{\Delta}_{rel}(s)$ as $\tilde{h}_{rel}(t)$. Then $||\Delta_{rel}(s)||_{\mathcal{H}_2}$ can be represented in the time domain as
\begin{align}
||\Delta_{rel}(s)||_{\mathcal{H}_{2}}&=\sqrt{\int_{0}^{\infty}trace\Big(\tilde{h}_{rel}(\tau)\tilde{h}_{rel}^T(\tau)\Big)d\tau}.\nonumber
\end{align}
Finally, $||\Delta_{rel}(s)||_{\mathcal{H}_{2,\tau}}$ can be written as
\begin{align}
||\Delta_{rel}(s)||_{\mathcal{H}_{2,\tau}}&=\sqrt{\int_{0}^{t_d}trace\Big(\tilde{h}_{rel}(\tau)\tilde{h}_{rel}^T(\tau)\Big)d\tau}.\nonumber
\end{align} It can readily be noted that $\hat{H}^{-1}(s)$ can be replaced with $\hat{G}^{-*}(s)$ without affecting $||\Delta_{rel}(s)||_{\mathcal{H}_{2,\tau}}$. The advantage, however, is that $||\Delta_{rel}(s)||_{\mathcal{H}_{2,\tau}}$ can be defined even when $\hat{H}(s)$ is not a stable minimum phase system, i.e., by replacing $\hat{H}^{-1}(s)$ with a stable $\hat{G}^{-*}(s)$. Next, we construct a state-space realization of $\hat{G}^{-*}(s)$ on similar lines with BST, wherein a similar state-space realization is constructed.

The observability gramian $\hat{Q}$ of the pair $(\hat{A},\hat{C})$ solves the following Lyapunov equation
   \begin{align}
   \hat{A}^T\hat{Q}+\hat{Q}\hat{A}+\hat{C}^T\hat{C}=0.\nonumber
   \end{align}
Let us define $\hat{B}_s$ and $\hat{A}_s$ as the following
  \begin{align}
  \hat{B}_s&=-\hat{Q}\hat{B}-\hat{C}^TD,&\hat{A}_s&=-\hat{A}-\hat{B}(D^TD)^{-1}\hat{B}_s^T.\nonumber
  \end{align}
  Let $\hat{X}_s$ solves the following Riccati equation
  \begin{align}
  \hat{A}_s\hat{X}_s+\hat{X}_s\hat{A}_s^T+\hat{X}_s\hat{B}_s(D^TD)^{-1}\hat{B}_s^T\hat{X}_s+\hat{B}(D^TD)^{-1}\hat{B}^T=0.\label{e37}
  \end{align}
  Then, a minimum phase realization of $\hat{G}^*(s)$ is given as
  \begin{align}
  \hat{A}_x&=-\hat{A}^T,&\hat{B}_x&=\hat{B}_s,&\hat{C}_x&=D^{-T}(\hat{B}^T-\hat{B}_s^T\hat{X}_s),&\hat{D}_x=D,\label{e38}
  \end{align}cf. \cite{zhou1996robust}. Further, a stable realization of $\hat{G}^{-*}(s)$ is given as
  \begin{align}
\hat{A}_{xi}&=\hat{A}_x-\hat{B}_x\hat{D}_x^{-1}\hat{C}_x,& \hat{B}_{xi}&=-\hat{B}_x\hat{D}_x^{-1}, & \hat{C}_{xi}&=\hat{D}_x^{-1}\hat{C}_x,& \hat{D}_{xi}&=\hat{D}_x^{-1},\label{e39}
\end{align}cf. \cite{zhou1995frequency,zhou1996robust}.

We are now in a position to state the algorithm. The proposed algorithm is named ``Time-limited Relative-error $\mathcal{H}_{2}$ MOR Algorithm (TLRHMORA)". The pseudo-code of TLRHMORA is given in Algorithm \ref{Alg1}. Step \ref{stp1} replaces $D$ with a full rank matrix if $D$ is rank deficient. Steps \ref{stp3}-\ref{stp4} compute a stable realization of $\hat{G}^{-*}(s)$. Steps \ref{stp5}-\ref{stp7} compute the reduction matrices. Steps \ref{stp8}-\ref{stp12} ensure the oblique projection condition $\hat{W}^T\hat{V}=I$ using the bi-orthogonal Gram-Schmidt method, cf. \cite{benner2011sparse}. Step \ref{stp13} updates the ROM.
\begin{algorithm}[!h]
\textbf{Input:} Original system: $(A,B,C,D)$; Desired time: $t_d$ sec; Initial guess: $(\hat{A},\hat{B},\hat{C})$; Allowable iteration: $k_{max}$.\\ \textbf{Output:} ROM $(\hat{A},\hat{B},\hat{C})$.
  \begin{algorithmic}[1]
      \STATE \textbf{if} ($rank[D]<m$) $D=\epsilon I$ \textbf{end if}\label{stp1}
      \STATE $k=0$, \textbf{while} (not converged) \textbf{do} $k=k+1$.
      \STATE Solve the equation (\ref{e37}) to compute $\hat{X}_s$.\label{stp3}
      \STATE Compute $(\hat{A}_{xi},\hat{B}_{xi},\hat{C}_{xi},\hat{D}_{xi})$ from (\ref{e38}) and (\ref{e39}).\label{stp4}
       \STATE Solve the equation (\ref{e6}) to compute $P_{12}$.\label{stp5}
       \STATE Solve the following equations to compute $E_1$ and $E_2$\\
$\hat{A}_{xi}E_1-E_1A+\hat{B}_{xi}Ce^{At_d}-e^{\hat{A}_{xi}t_d}\hat{B}_{xi}C=0$,\\
$\hat{A}_{xi}E_2-E_2\hat{A}+\hat{B}_{xi}\hat{C}e^{\hat{A}t_d}-e^{\hat{A}_{xi}t_d}\hat{B}_{xi}\hat{C}=0$.\\
      \STATE Compute $\tilde{Q}_{33}$, $\tilde{Q}_{13}$, $\tilde{Q}_{23}$, and $\tilde{Q}_{12}$ by solving\label{stp7}\\
      $\hat{A}_{xi}^T\tilde{Q}_{33}+\tilde{Q}_{33}\hat{A}_{xi}+\hat{C}_{xi}^T\hat{C}_{xi}-e^{A_{xi}^Tt_d}\hat{C}_{xi}^T\hat{C}_{xi}e^{A_{xi}t_d}=0$,\\
      $A^T\tilde{Q}_{13}+\tilde{Q}_{13}\hat{A}_{xi}+C^T\hat{B}_{xi}^T\tilde{Q}_{33}+C^T\hat{D}_{xi}^T\hat{C}_{xi}-e^{A^Tt_d}C^TD_{xi}^TC_{xi}e^{A_{xi}t_d}-E_1^TC_{xi}^TC_{xi}e^{A_{xi}t_d}=0$,\\
      $\hat{A}^T\tilde{Q}_{23}+\tilde{Q}_{23}\hat{A}_{xi}-\hat{C}^T\hat{B}_{xi}^T\tilde{Q}_{33}+\hat{C}^T\hat{D}_{xi}^T\hat{C}_{xi}-e^{\hat{A}^Tt_d}\hat{C}^TD_{xi}^TC_{xi}e^{A_{xi}t_d}-E_2^TC_{xi}^TC_{xi}e^{A_{xi}t_d}=0$,\\
      $A^T\tilde{Q}_{12}+\tilde{Q}_{12}\hat{A}+C^T\hat{B}_{xi}^T\tilde{Q}_{23}^T-\tilde{Q}_{13}\hat{B}_{xi}\hat{C}-C^T\hat{D}_{xi}^T\hat{D}_{xi}\hat{C}+e^{A^Tt_d}C^TD_{xi}^TD_{xi}\hat{C}e^{\hat{A}t_d}+E_1^TC_{xi}^TD_{xi}\hat{C}e^{\hat{A}t_d}-e^{A^Tt_d}C^TD_{xi}^TC_{xi}E_2-E_1^TC_{xi}^TC_{xi}E_2=0$.
      \STATE \textbf{for} $i=1,\ldots,r$ \textbf{do}\label{stp8}
      \STATE $v=P_{12}(:,i)$, $v=\prod_{k=1}^{i}\big(I-P_{12}(:,k)\tilde{Q}_{12}(:,k)^T\big)v$.
      \STATE $w=\tilde{Q}_{12}(:,i)$, $w=\prod_{k=1}^{i}\big(I-\tilde{Q}_{12}(:,k)P_{12}(:,k)^T\big)w$.
      \STATE $v=\frac{v}{||v||_2}$, $w=\frac{w}{||w||_2}$, $v=\frac{v}{w^Tv}$, $\hat{V}(:,i)=v$, $\hat{W}(:,i)=w$.
      \STATE \textbf{end for}\label{stp12}
      \STATE $\hat{A}=\hat{W}^TA\hat{V}$, $\hat{B}=\hat{W}^TB$, $\hat{C}=C\hat{V}$.\label{stp13}
      \STATE \textbf{if} ($k=k_{max}$) Break loop \textbf{end if}
      \STATE \textbf{end while}
  \end{algorithmic}
  \caption{TLRHMORA}\label{Alg1}
\end{algorithm}
\subsubsection{General Time Interval}
So far, we restricted ourselves for simplicity to the case when the desired time interval starts from $0$ sec. The problem can be generalized to any time interval $[t_1,t_2]$ sec by making a few changes. Note that for the general time interval case, $P_{rel}$ and $Q_{rel}$ solve the following Lyapunov equations
\begin{align}
A_{rel}P_{rel}+P_{rel}A_{rel}^T+e^{A_{rel}t_1}B_{rel}B_{rel}^Te^{A_{rel}^Tt_1}-e^{A_{rel}t_2}B_{rel}B_{rel}^Te^{A_{rel}^Tt_2}&=0,\nonumber\\
A_{rel}^TQ_{rel}+Q_{rel}A_{rel}+e^{A_{rel}^Tt_1}C_{rel}^TC_{rel}e^{A_{rel}t_1}-e^{A_{rel}^Tt_2}C_{rel}^TC_{rel}e^{A_{rel}t_2}&=0,\nonumber
\end{align}cf. \cite{gawronski1990model}.
Accordingly, the reduction matrices in TLRHMORA for the general time interval case can be computed by solving the following equations in each iteration
\begin{align}
AP_{12}+P_{12}\hat{A}^T+e^{At_1}B\hat{B}^Te^{\hat{A}^Tt_1}-e^{At_2}B\hat{B}^Te^{\hat{A}^Tt_2}&=0,\nonumber\\
      \hat{A}_{xi}^T\tilde{Q}_{33}+\tilde{Q}_{33}\hat{A}_{xi}+e^{A_{xi}^Tt_1}\hat{C}_{xi}^T\hat{C}_{xi}e^{A_{xi}t_1}-e^{A_{xi}^Tt_2}\hat{C}_{xi}^T\hat{C}_{xi}e^{A_{xi}t_2}&=0,\nonumber\\
      A^T\tilde{Q}_{13}+\tilde{Q}_{13}\hat{A}_{xi}+C^T\hat{B}_{xi}^T\tilde{Q}_{33}+e^{A^Tt_1}C^TD_{xi}^TC_{xi}e^{A_{xi}t_1}+E_{1,t_1}^TC_{xi}^TC_{xi}e^{A_{xi}t_1}&\nonumber\\
      -e^{A^Tt_2}C^TD_{xi}^TC_{xi}e^{A_{xi}t_2}-E_{1,t_2}^TC_{xi}^TC_{xi}e^{A_{xi}t_2}&=0,\nonumber\\
      \hat{A}^T\tilde{Q}_{23}+\tilde{Q}_{23}\hat{A}_{xi}-\hat{C}^T\hat{B}_{xi}^T\tilde{Q}_{33}+e^{\hat{A}^Tt_1}\hat{C}^TD_{xi}^TC_{xi}e^{A_{xi}t_1}+E_{2,t_1}^TC_{xi}^TC_{xi}e^{A_{xi}t_1}&\nonumber\\
      -e^{\hat{A}^Tt_2}\hat{C}^TD_{xi}^TC_{xi}e^{A_{xi}t_2}-E_{2,t_2}^TC_{xi}^TC_{xi}e^{A_{xi}t_2}&=0,\nonumber\\
      A^T\tilde{Q}_{12}+\tilde{Q}_{12}\hat{A}+C^T\hat{B}_{xi}^T\tilde{Q}_{23}^T-\tilde{Q}_{13}\hat{B}_{xi}\hat{C}-e^{A^Tt_1}C^TD_{xi}^TD_{xi}\hat{C}e^{\hat{A}t_1}&\nonumber\\
      -E_{1,t_1}^TC_{xi}^TD_{xi}\hat{C}e^{\hat{A}t_1}+e^{A^Tt_1}C^TD_{xi}^TC_{xi}E_{2,t_1}&\nonumber\\
      +E_{1,t_1}^TC_{xi}^TC_{xi}E_{2,t_1}+e^{A^Tt_2}C^TD_{xi}^TD_{xi}\hat{C}e^{\hat{A}t_2}&\nonumber\\
      +E_{1,t_2}^TC_{xi}^TD_{xi}\hat{C}e^{\hat{A}t_2}-e^{A^Tt_2}C^TD_{xi}^TC_{xi}E_{2,t_2}&\nonumber\\
      -E_{1,t_2}^TC_{xi}^TC_{xi}E_{2,t_2}&=0,\nonumber
      \end{align}
      wherein $E_{1,t}$ and $E_{2,t}$ can be computed by solving the following Sylvester equations
      \begin{align}
A_{xi}E_{1,t}-E_{1,t}A+B_{xi}Ce^{At}-e^{A_{xi}t}B_{xi}C&=0,\nonumber\\
A_{xi}E_{2,t}-E_{2,t}\hat{A}+B_{xi}\hat{C}e^{\hat{A}t}-e^{A_{xi}t}B_{xi}\hat{C}&=0.\nonumber
\end{align}
\subsubsection{Stopping Criterion}
Most of the $\mathcal{H}_2$ MOR algorithms are iterative methods with no guarantee of convergence. The same is the case with TLRHMORA. Further, it has been already discussed that even if TLRHMORA converges, it generally cannot achieve a local optimum due to the inherent construction of the oblique projection framework. Therefore, it is reasonable to stop the algorithm prematurely if it does not converge within an allowable number of iterations. This stopping criterion can help in keeping the computational cost in check, especially in a large-scale setting.
\subsubsection{Selection of the Initial Guess}
Iterative methods generally improve as the number of iterations progresses. However, a good initial guess can significantly improve the final outcome in terms of accuracy and convergence. It is established in \cite{anic2013interpolatory,breiten2015near,zulfiqar2022frequency} that to ensure a small weighted additive error in the $\mathcal{H}_{2}$ norm, the initial guess should have the dominant poles (with large residues) of the original model and the frequency weight. By following this guideline, the initial guess for TLRHMORA should have dominant poles and/or zeros of the original model. Such an initial guess can be constructed by using computationally efficient algorithms proposed in \cite{rommes2006efficient} and \cite{martins2007computation}. Although there is no guarantee that the final ROM will preserve these poles and zeros, it ensures that $||\Delta_{rel}(s)||_{\mathcal{H}_{2,\tau}}$ is small to start with.
\subsubsection{Computational Aspects}
The computation of $\hat{X}_s$ and $(\hat{A}_{xi},\hat{B}_{xi},\hat{C}_{xi},\hat{D}_{xi})$ can be done cheaply, as it only involves small matrices. Similarly, the matrix exponentials $e^{\hat{A}t}$, $e^{\hat{A}_{xi}t}$, and $E_2$ can also be computed cheaply, as it involves small matrices. Further, the computation of $\tilde{Q}_{33}$ and $\tilde{Q}_{23}$ is again cheap owing to the small-sized matrices involved. However, the matrix exponential $e^{At}$ can be a computational challenge if $A$ is a large-scale matrix. Note that we need this matrix exponential's products $e^{At}B$ and $Ce^{At}$ for the computation of $P_{12}$, $E_1$, $\tilde{Q}_{13}$, and $\tilde{Q}_{12}$. In \cite{kurschner2018balanced}, computationally efficient Krylov subspace-based projection methods are proposed that can accurately approximate $e^{At}B$ and $Ce^{At}$ in a large-scale setting as $\hat{V}e^{\hat{A}t}\hat{B}$ and $\hat{C}e^{\hat{A}t}\hat{W}^T$, respectively. When $A$ is a large-scale matrix, these methods can be used to keep TLRHMORA computationally viable.

Once the matrix exponential's products $e^{At}B$ and $Ce^{At}$ are computed, the most computationally demanding step in TLRHOMRA is to compute $P_{12}$, $E_1$, $\tilde{Q}_{13}$, and $\tilde{Q}_{12}$ by solving their respective Sylvester equations. These Sylvester equations have the following special structure
\begin{align}
KJ+JL+MN=0,\nonumber
\end{align} wherein $K\in\mathbb{R}^{n\times n}$, $L\in\mathbb{R}^{r\times r}$, $M\in\mathbb{R}^{n\times q}$, $N\in\mathbb{R}^{q\times r}$, and $q\ll n$. The big matrices $K$ and $M$ in this type of Sylvester equation are sparse because the mathematically modeling of most modern-day systems ends up in a large-scale but sparse state-space model. The small matrices $L$ and $N$ are dense; hence, giving this equation the name ``\textit{sparse-dense} Sylvester equation". It frequently arises in $\mathcal{H}_2$ MOR algorithms, and an efficient solution for this type of Sylvester equation is proposed in \cite{benner2011sparse}. It is highlighted in \cite{benner2011sparse} that the computational time required to solve this equation can be reduced further by using sparse linear matrix solvers, like \cite{davis2004algorithm} and \cite{demmel1999supernodal}. Further, it is noted in \cite{wolf2014h} that the computational cost of solving a ``\textit{sparse-dense} Sylvester equation" is almost the same as that of constructing a rational Krylov subspace, which is well known to be computationally efficient in a large-scale setting; also see \cite{panzer2014model}. In short, the solution of a ``\textit{sparse-dense} Sylvester equation" is viable even in a large-scale setting due to the availability of efficient solvers in the literature. The remaining steps in TLRHMORA only involve simple operations that can be executed cheaply. Moreover, by using a maximum number of allowable iterations as a stopping criterion, the number of times these Sylvester equations need to be solved can be limited in case TLRHMORA does not converge quickly.

To conclude, by approximating the matrix exponential's products $e^{{A}t}B$ and $Ce^{{A}t}$ using the Krylov subspace methods proposed in \cite{kurschner2018balanced} and solving the ``\textit{sparse-dense} Sylvester equation" using the efficient solver proposed in \cite{benner2011sparse}, TLRHMORA remains a computationally tractable algorithm in a large-scale setting.
\section{Numerical Results}
In this section, TLRHMORA is applied to three models taken from the benchmark collection for testing MOR algorithms, cf. \cite{chahlaoui2005benchmark}, and its performance is compared with TLBT, TLBST, and TLIRKA. The first two models are SISO systems, while the third model is a MIMO system. All the tests are performed using MATLAB R2016a on a laptop with 16GB memory and a 2GHZ Intel i7 processor. All the Lyapunov and Sylvester equations are solved using MATLAB's \textit{``lyap"} command. The Riccati equations are solved using MATLAB's \textit{``care"} command. The matrix exponentials are computed using MATLAB's \textit{``expm()"} command. The maximum number of iterations in TLRHMORA is set to 50. Since all the models considered in this section have zero $D$-matrix, it is replaced with $D=0.0001I$ in TLBST and TLRHMORA. For fairness, TLIRKA and TLRHMORA are initialized with the same arbitrary initial guess. The desired time interval is chosen arbitrarily for demonstration purposes.
\subsection{Clamped Beam}
The clamped beam model is a $348^{th}$ order SISO system taken from the benchmark collection of \cite{chahlaoui2005benchmark}. The desired time interval in this example is selected as $[0,0.5]$ sec for demonstration purposes. ROMs of orders $5-10$ are generated using TLBT, TLBST, TLIRKA, and TLRHMORA. The relative errors $||\Delta_{rel}(s)||_{\mathcal{H}_{2,\tau}}$ are tabulated in Table \ref{tab1}, and it can be seen that TLRHMORA offers supreme accuracy.
\begin{table}[!h]
\centering
\caption{$||\Delta_{rel}(s)||_{\mathcal{H}_{2,\tau}}$}\label{tab1}
\begin{tabular}{|c|c|c|c|c|}
\hline
Order & TLBT & TLBST     & TLIRKA     & TLRHMORA \\ \hline
5     & 56.6676 & 43.4164 & 45.0637 & 22.7069 \\ \hline
6     & 56.3121& 75.4351 & 38.5722  & 9.4143 \\ \hline
7     & 15.6814& 10.4421 & 14.6589  & 5.5747 \\ \hline
8     & 5.2759   & 26.6056 & 4.5883  & 2.0863  \\ \hline
9     & 6.3876   & 17.7543  & 7.1781  & 1.9531 \\ \hline
10    & 0.5775   & 1.2804  & 0.5691   & 0.5256  \\ \hline
\end{tabular}
\end{table}
To visually analyze what these numbers mean in terms of accuracy, the impulse responses of $\Delta_{add}(s)$ for $6^{th}$ order ROMs within the desired time interval $[0,0.5]$ sec are plotted in Figure \ref{fig1}.
\begin{figure}[!h]
  \centering
  \includegraphics[width=12cm]{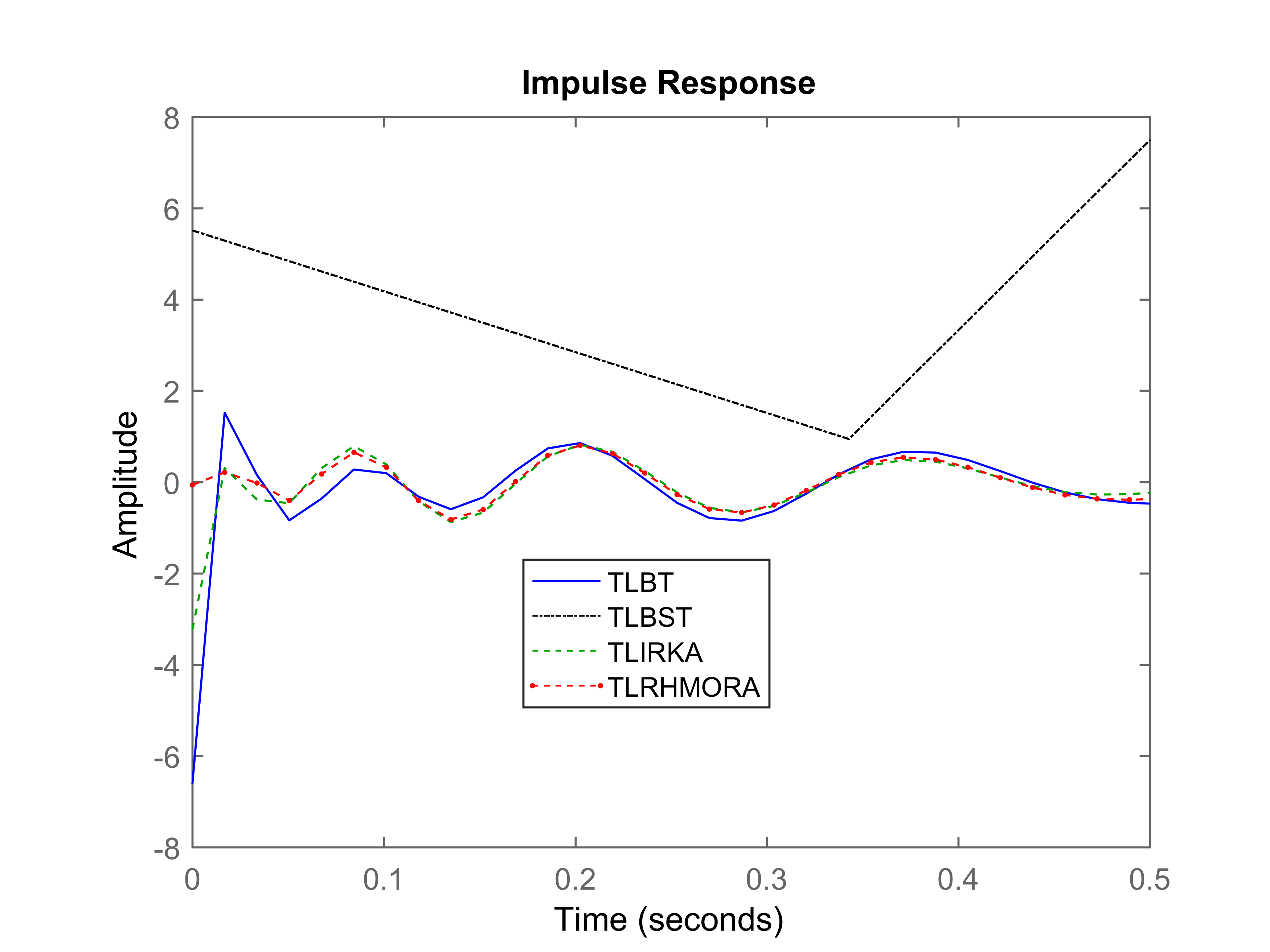}
  \caption{Impulse response of $\Delta_{add}(s)$ within $[0,0.5]$ sec}\label{fig1}
\end{figure} It can be seen that TLRHMORA offers the best performance. The highest error in this case is incurred by TLBST, and the impulse response plot corresponds accordingly to the numerical error given in Table \ref{tab1}. TLBT, TLIRKA, and TLRHMORA, in this case, have offered comparable accuracy for the most part of the desired time interval. However, TLBT and TLIRKA have incurred large numerical errors near 0 sec. This is the reason why $||\Delta_{rel}(s)||_{\mathcal{H}_{2,\tau}}$ for TLBT and TIRKA is significantly larger than that for TLRHMORA. Note, TLIRKA and TLRHMORA are both $\mathcal{H}_{2,\tau}$ MOR algorithms but with different cost functions, i.e., TLIRKA minimizes $||\Delta_{add}(s)||_{\mathcal{H}_{2,\tau}}$, while TLRHMORA minimizes $||\Delta_{rel}(s)||_{\mathcal{H}_{2,\tau}}$. Although both algorithms are initialized with the same initial guess, it can be seen in Table \ref{tab1} and Figure \ref{fig1} that TLRHMORA ensures superior accuracy due to its cost function.
\subsection{Artificial Dynamical System}
The artificial dynamical model is a $1006^{th}$ order SISO system taken from the benchmark collection of \cite{chahlaoui2005benchmark}. The desired time interval in this example is selected as $[0,1]$ sec for demonstration purposes. ROMs of orders $11-15$ are generated using TLBT, TLBST, TLIRKA, and TLRHMORA. The relative errors $||\Delta_{rel}(s)||_{\mathcal{H}_{2,\tau}}$ are tabulated in Table \ref{tab2}, and it can be seen that TLRHMORA offers supreme accuracy.
\begin{table}[!h]
\centering
\caption{$||\Delta_{rel}(s)||_{\mathcal{H}_{2,\tau}}$}\label{tab2}
\begin{tabular}{|c|c|c|c|c|}
\hline
Order & TLBT & TLBST     & TLIRKA     & TLRHMORA \\ \hline
11     & 0.8117 & 0.1098 & 0.5251 & 0.0488 \\ \hline
12     & 0.2270& 0.2858 & 0.1885  & 0.0511 \\ \hline
13     & 0.2250& 4.9610 & 4.5675  & 0.0317 \\ \hline
14     & 0.1760   & 0.2748 & 2.6495  & 0.0218  \\ \hline
15     & 0.1392   & 0.1342  & 0.8801  & 0.0188 \\ \hline
\end{tabular}
\end{table}
To visually analyze what these numbers mean in terms of accuracy, the impulse responses of $\Delta_{add}(s)$ for $11^{th}$ order ROMs within the desired time interval $[0,1]$ sec are plotted in Figure \ref{fig2}.
\begin{figure}[!h]
  \centering
  \includegraphics[width=12cm]{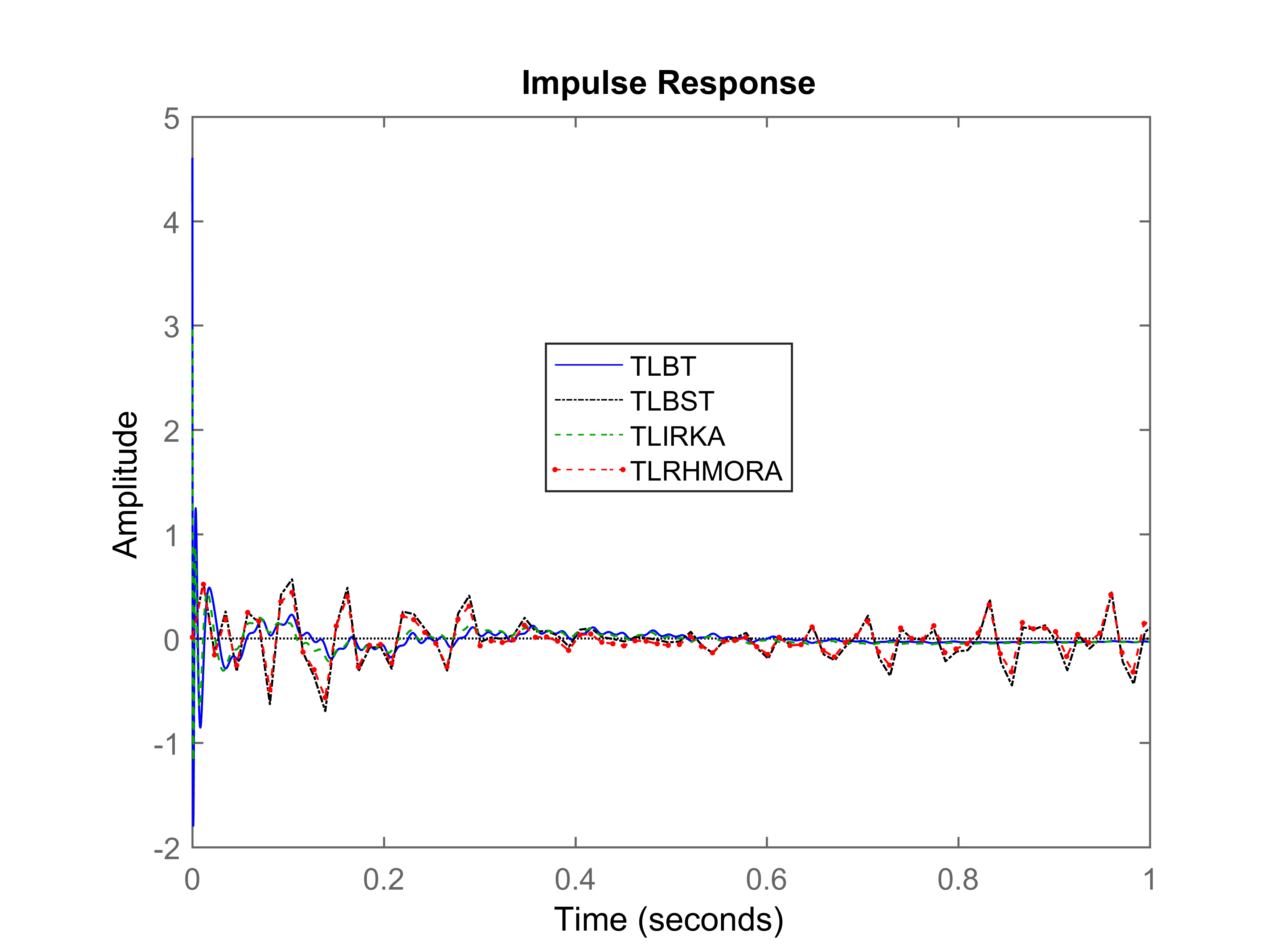}
  \caption{Impulse response of $\Delta_{add}(s)$ within $[0,1]$ sec}\label{fig2}
\end{figure} It can be seen that TLRHMORA offers the best performance. TLBT, TLBST, TLIRKA, and TLRHMORA, in this case, have offered comparable accuracy for the most part of the desired time interval. However, TLBT, TLBST, and TLIRKA have incurred large numerical errors near 0 sec, which can be seen in Figure \ref{fig3} (a magnified version of Figure \ref{fig2}).
\begin{figure}[!h]
  \centering
  \includegraphics[width=12cm]{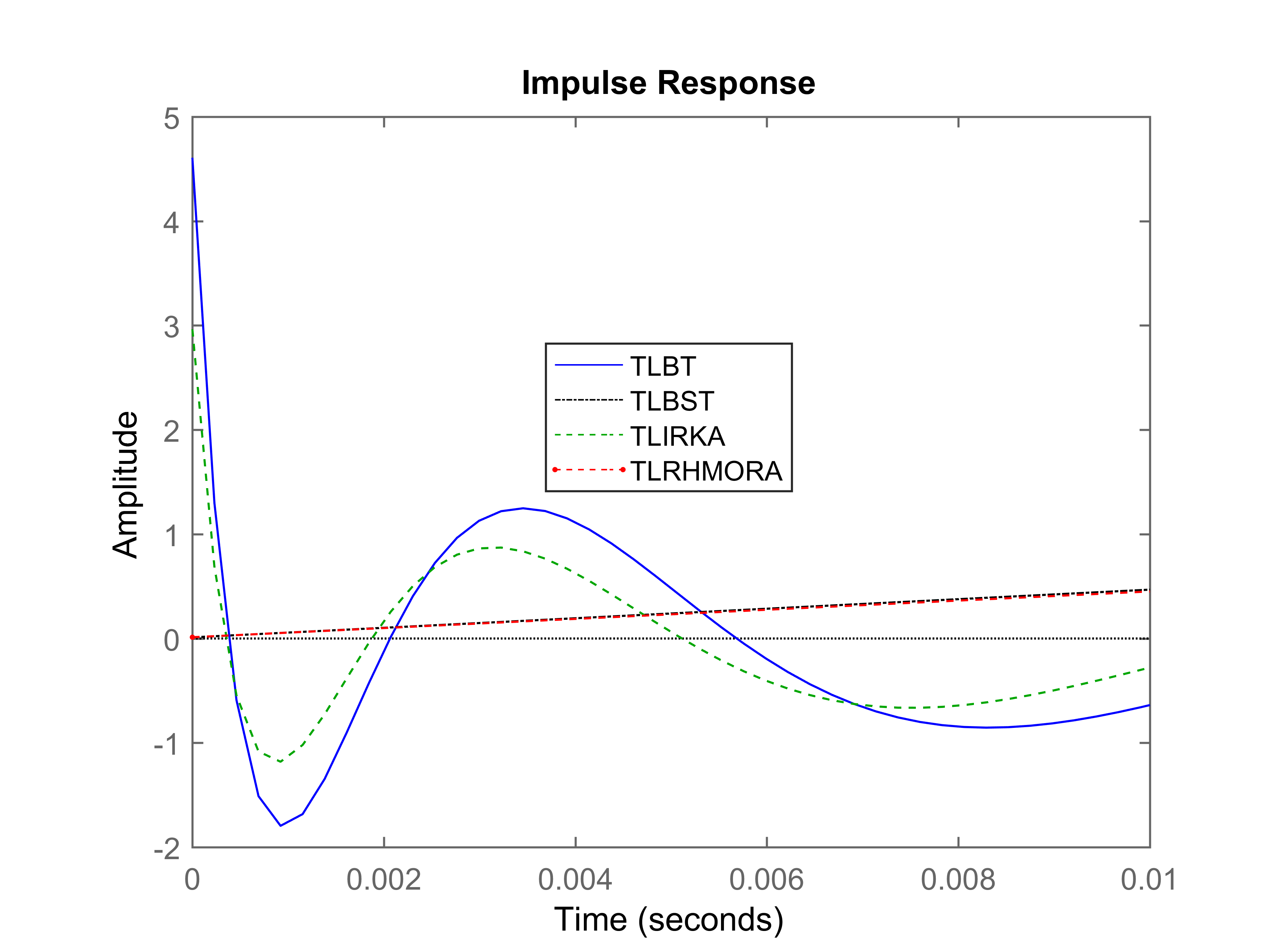}
  \caption{Impulse response of $\Delta_{add}(s)$ within $[0,0.01]$ sec}\label{fig3}
\end{figure}
This is the reason why $||\Delta_{rel}(s)||_{\mathcal{H}_{2,\tau}}$ for TLBT, TLBST, and TIRKA is significantly larger than that for TLRHMORA.
\subsection{International Space Station}
The international space station model is a $270^{th}$ order MIMO system taken from the benchmark collection of \cite{chahlaoui2005benchmark}. It has three inputs and three outputs. The desired time interval in this example is selected as $[0,2]$ sec for demonstration purposes. ROMs of orders $5-9$ are generated using TLBT, TLBST, TLIRKA, and TLRHMORA. The relative errors $||\Delta_{rel}(s)||_{\mathcal{H}_{2,\tau}}$ are tabulated in Table \ref{tab3}, and it can be seen that TLRHMORA offers supreme accuracy.
\begin{table}[!h]
\centering
\caption{$||\Delta_{rel}(s)||_{\mathcal{H}_{2,\tau}}$}\label{tab3}
\begin{tabular}{|c|c|c|c|c|}
\hline
Order & TLBT & TLBST     & TLIRKA     & TLRHMORA \\ \hline
5     & 2.0787 & 2.1599 & 2.4186 & 1.9615 \\ \hline
6     & 2.0865& 1.9597 & 1.9485 & 1.9367 \\ \hline
7     & 2.0040& 1.9317 & 2.5756  & 1.5787 \\ \hline
8     & 1.6106   & 1.8876 & 1.8760 & 1.5636  \\ \hline
9     & 1.6166   & 1.8859  & 1.9060  & 1.0615 \\ \hline
\end{tabular}
\end{table} There are nine impulse response plots of $\Delta_{add}(s)$ for this example. Therefore, we refrain from showing these for the economy of space.
\section{Conclusion}
The problem of time-limited $\mathcal{H}_2$ MOR based on relative error expression is considered. The necessary conditions for a local optimum are derived for the case when the ROM is a stable minimum phase system. Based on these conditions, the reduction matrices for an oblique projection algorithm are identified. Then, a generalized iterative algorithm, which does not require the ROM to be a stable minimum phase system in every iteration, is proposed. Unlike TLBST, the proposed algorithm does not require the solutions of large-scale Lyapunov and Riccati equations. Instead, solutions of small-scale Lyapunov and Riccati equations are required that can be computed cheaply. The numerical results confirm that the proposed algorithm is accurate and offers high fidelity.
\section*{Appendix}
Let us define the cost function $J$ as $J=||\Delta_{rel}||_{\mathcal{H}_{2,\tau}}^2$ and denote the first-order derivatives of $Q_{11}$, $Q_{12}$, $Q_{13}$, $Q_{22}$, $Q_{23}$, $Q_{33}$, $E_1$, and $J$ with respect to $\hat{A}$ as $\Delta_{Q_{11}}^{\hat{A}}$, $\Delta_{Q_{12}}^{\hat{A}}$, $\Delta_{Q_{13}}^{\hat{A}}$, $\Delta_{Q_{22}}^{\hat{A}}$, $\Delta_{Q_{23}}^{\hat{A}}$, $\Delta_{Q_{33}}^{\hat{A}}$, $\Delta_{E_1}^{\hat{A}}$, and $\Delta_{J}^{\hat{A}}$, respectively. Further, let us denote the differential of $\hat{A}$ as $\Delta_{\hat{A}}$. By differentiating $J$ with respect to $\hat{A}$, we get
      \begin{align}
      \Delta_{J}^{\hat{A}}&=trace(B^T\Delta_{Q_{11}}^{\hat{A}}B+2B^T\Delta_{Q_{12}}^{\hat{A}}\hat{B}+\hat{B}^T\Delta_{Q_{22}}^{\hat{A}}\hat{B})\nonumber\\
      &=trace(BB^T\Delta_{Q_{11}}^{\hat{A}}+2B\hat{B}^T(\Delta_{Q_{12}}^{\hat{A}})^T+\hat{B}\hat{B}^T\Delta_{Q_{22}}^{\hat{A}}).\nonumber
      \end{align}
By differentiating the equations (\ref{e16}), (\ref{e17}), and (\ref{e19}) with respect to $\hat{A}$, we get
      \begin{align}
      A^T\Delta_{Q_{11}}^{\hat{A}}+\Delta_{Q_{11}}^{\hat{A}}A+R_1&=0,\hspace*{2cm}\label{e40}\\
      A^T\Delta_{Q_{12}}^{\hat{A}}+\Delta_{Q_{12}}^{\hat{A}}\hat{A}+R_2&=0,\label{e41}\\
      \hat{A}^T\Delta_{Q_{22}}^{\hat{A}}+\Delta_{Q_{22}}^{\hat{A}}\hat{A}+R_4&=0,\label{e42}
      \end{align}
      wherein
      \begin{align}
      R_1&=C^TB_i^T(\Delta_{Q_{13}}^{\hat{A}})^T+\Delta_{Q_{13}}^{\hat{A}}B_iC-(\Delta_{E_1}^{\hat{A}})^TC_i^TD_iCe^{At_d}-e^{A^Tt_d}C^TD_i^TC_i\Delta_{E_1}^{\hat{A}}\nonumber\\
      &\hspace*{1cm}-(\Delta_{E_1}^{\hat{A}})^TC_i^TC_iE_1-E_1^TC_i^TC_i\Delta_{E_1}^{\hat{A}}\nonumber\\
      R_2&=Q_{12}\Delta_{\hat{A}}+C^TB_i^T\Delta_{Q_{23}}^{\hat{A}}-\Delta_{Q_{13}}^{\hat{A}}B_i\hat{C}+t_de^{A^Tt_d}C^TD_i^TC_ie^{\hat{A}t_d}\Delta_{\hat{A}}\nonumber\\
      &\hspace*{1cm}+(\Delta_{E_1}^{\hat{A}})^TC_i^TC_ie^{\hat{A}t_d}+t_dE_1^TC_i^TC_ie^{\hat{A}t_d}\Delta_{\hat{A}}-t_de^{A^Tt_d}C^TD_i^TC_iE_2\Delta_{\hat{A}}\nonumber\\
      &\hspace*{1cm}-(\Delta_{E_1}^{\hat{A}})^TC_i^TC_iE_2-t_dE_1^TC_i^TC_iE_2\Delta_{\hat{A}}\nonumber\\
      R_4&=(\Delta_{\hat{A}})^TQ_{22}+Q_{22}\Delta_{\hat{A}}-\hat{C}^TB_i^T\Delta_{Q_{23}}^{\hat{A}}-\Delta_{Q_{23}}^{\hat{A}}B_i\hat{C}-t_d(\Delta_{\hat{A}})^Te^{\hat{A}^Tt_d}C_i^TC_ie^{\hat{A}t_d}\nonumber\\
      &\hspace*{1cm}-t_de^{\hat{A}^Tt_d}C_i^TC_ie^{\hat{A}t_d}\Delta_{\hat{A}}+t_d(\Delta_{\hat{A}})^TE_2^TC_i^TC_ie^{\hat{A}t_d}+t_dE_2^TC_i^TC_ie^{\hat{A}t_d}\Delta_{\hat{A}}\nonumber\\
      &\hspace*{1cm}+t_d(\Delta_{\hat{A}})^Te^{\hat{A}^Tt_d}C_i^TC_iE_2+t_de^{\hat{A}^Tt_d}C_i^TC_iE_2\Delta_{\hat{A}}-t_d(\Delta_{\hat{A}})^TE_2^TC_i^TC_iE_2\nonumber\\
      &\hspace*{1cm}-t_dE_2^TC_i^TC_iE_2\Delta_{\hat{A}}.\nonumber
      \end{align}
      By applying Lemma \ref{lemma} on the equations (\ref{e40}) and (\ref{e23}), (\ref{e41}) and (\ref{e24}), and (\ref{e42}) and (\ref{e26}), we get
      \begin{align}
      trace(BB^T\Delta_{Q_{11}}^{\hat{A}})&=trace(R_1X_{11}),\nonumber\\
      trace\big(B\hat{B}^T(\Delta_{Q_{12}}^{\hat{A}})^T\big)&=trace(R_2^TX_{12}),\nonumber\\
      trace(\hat{B}\hat{B}^T\Delta_{Q_{22}}^{\hat{A}})&=trace(R_4X_{22}).\nonumber
      \end{align}
      Thus
      \begin{align}
      \Delta_{J}^{\hat{A}}&=trace\Big(2Q_{12}^TX_{12}(\Delta_{\hat{A}})^T+2Q_{22}X_{22}(\Delta_{\hat{A}})^T+2t_de^{\hat{A}^Tt_d}C_i^TD_iCe^{At_d}X_{12}(\Delta_{\hat{A}})^T\nonumber\\
      &\hspace*{0.5cm}-2t_de^{\hat{A}^Tt_d}C_i^TC_ie^{\hat{A}t_d}X_{22}(\Delta_{\hat{A}})^T+2t_de^{\hat{A}^Tt_d}C_i^TC_iE_1X_{12}(\Delta_{\hat{A}})^T\nonumber\\
      &\hspace*{0.5cm}+2t_de^{\hat{A}^Tt_d}C_i^TC_iE_2X_{22}(\Delta_{\hat{A}})^T-2t_dE_2^TC_i^TD_iCe^{At_d}X_{12}(\Delta_{\hat{A}})^T\nonumber\\
      &\hspace*{0.5cm}+2t_dE_2^TC_i^TC_ie^{\hat{A}t_d}X_{22}(\Delta_{\hat{A}})^T-2t_dE_2^TC_i^TC_iE_1X_{12}(\Delta_{\hat{A}})^T\nonumber\\
      &\hspace*{0.5cm}-2t_dE_2^TC_i^TC_iE_2X_{22}(\Delta_{\hat{A}})^T+2X_{11}C^TB_i^T(\Delta_{Q_{13}}^{\hat{A}})^T-2X_{12}\hat{C}^TB_i^T(\Delta_{Q_{13}}^{\hat{A}})^T\nonumber\\
      &\hspace*{0.5cm}+2X_{12}^TC^TB_i^T\Delta_{Q_{23}}^{\hat{A}}-2X_{22}\hat{C}^TB_i^T\Delta_{Q_{23}}^{\hat{A}}-2C_i^TC_iE_1X_{11}(\Delta_{E_1}^{\hat{A}})^T\nonumber\\
      &\hspace*{0.5cm}-2C_i^TC_iE_2X_{12}^T(\Delta_{E_1}^{\hat{A}})^T-2C_i^TD_iCe^{At_d}X_{11}(\Delta_{E_1}^{\hat{A}})^T\nonumber\\
      &\hspace*{0.5cm}+2C_i^TC_ie^{\hat{A}t_d}X_{12}^T(\Delta_{E_1}^{\hat{A}})^T\Big).\nonumber
      \end{align}
      By differentiating the equations (\ref{e18}), (\ref{e20}), (\ref{e21}), and (\ref{e8}) with respect to $\hat{A}$, we get
      \begin{align}
      A^T\Delta_{Q_{13}}^{\hat{A}}+\Delta_{Q_{13}}^{\hat{A}}A_i+R_3&=0,\label{e43}\\
      \hat{A}^T\Delta_{Q_{23}}^{\hat{A}}+\Delta_{Q_{23}}^{\hat{A}}A_i+R_5&=0,\label{e44}\\
      A_i^T\Delta_{Q_{33}}^{\hat{A}}+\Delta_{Q_{33}}^{\hat{A}}A_i+R_6&=0,\label{e45}\\
      A_i\Delta_{E_1}^{\hat{A}}-\Delta_{E_1}^{\hat{A}}A+R_7&=0,\label{e46}
      \end{align}
      wherein
      \begin{align}
      R_3&=Q_{13}\Delta_{\hat{A}}+C^TB_i^T\Delta_{Q_{33}}^{\hat{A}}-t_de^{A^Tt_d}C^TD_i^TC_ie^{A_it_d}\Delta_{\hat{A}}-t_dE_1^TC_i^TC_ie^{A_it_d}\Delta_{\hat{A}}\nonumber\\
      &\hspace*{0.5cm}-(\Delta_{E_1}^{\hat{A}})^TC_i^TC_ie^{A_it_d},\nonumber\\
      R_5&=(\Delta_{\hat{A}})^TQ_{23}+Q_{23}\Delta_{\hat{A}}-\hat{C}^TB_i^T\Delta_{Q_{33}}^{\hat{A}}+t_d(\Delta_{\hat{A}})^Te^{\hat{A}^Tt_d}C_i^TC_ie^{A_it_d}\nonumber\\
      &\hspace*{0.5cm}+t_de^{\hat{A}^Tt_d}C_i^TC_ie^{A_it_d}\Delta_{\hat{A}}-t_dE_2^TC_i^TC_ie^{A_it_d}\Delta_{\hat{A}}-t_d(\Delta_{\hat{A}})^TE_2^TC_i^TC_ie^{A_it_d},\nonumber\\
      R_6&=(\Delta_{\hat{A}})^TQ_{33}+Q_{33}\Delta_{\hat{A}}-t_de^{A_i^Tt_d}C_i^TC_ie^{A_it_d}\Delta_{\hat{A}}-t_d(\Delta_{\hat{A}})^Te^{A_i^Tt_d}C_i^TC_ie^{A_it_d},\nonumber\\
      R_7&=\Delta_{\hat{A}}E_1-t_de^{A_it_d}\Delta_{\hat{A}}B_iC,\nonumber
      \end{align}
      By applying Lemma \ref{lemma} on the equations (\ref{e43}) and (\ref{e25}), and (\ref{e44}) and (\ref{e27}), we get
      \begin{align}
      trace(R_3^TX_{13})&=trace\big(X_{11}C^TB_i^T(\Delta_{Q_{13}}^{\hat{A}})^T-X_{12}\hat{C}^TB_i^T(\Delta_{Q_{13}}^{\hat{A}})^T\big),\nonumber\\
      trace(R_5^TX_{23})&=trace\big(X_{12}^TC^TB_i^T\Delta_{Q_{23}}^{\hat{A}}-X_{22}\hat{C}^TB_i\Delta_{Q_{23}}^{\hat{A}}\big).\nonumber
      \end{align}
      Thus
      \begin{align}
      \Delta_{J}^{\hat{A}}&=trace\Big(2Q_{12}^TX_{12}(\Delta_{\hat{A}})^T+2Q_{22}X_{22}(\Delta_{\hat{A}})^T+2Q_{13}^TX_{13}(\Delta_{\hat{A}})^T\nonumber\\
      &\hspace*{0.5cm}+2Q_{23}X_{23}(\Delta_{\hat{A}})^T+2Q_{23}X_{23}^T(\Delta_{\hat{A}})^T+2t_de^{\hat{A}^Tt_d}C_i^TD_iCe^{At_d}X_{12}(\Delta_{\hat{A}})^T\nonumber\\
      &\hspace*{0.5cm}-2t_de^{\hat{A}^Tt_d}C_i^TC_ie^{\hat{A}t_d}X_{22}(\Delta_{\hat{A}})^T+2t_de^{\hat{A}^Tt_d}C_i^TC_iE_1X_{12}(\Delta_{\hat{A}})^T\nonumber\\
      &\hspace*{0.5cm}+2t_de^{\hat{A}^Tt_d}C_i^TC_iE_2X_{22}(\Delta_{\hat{A}})^T-2t_dE_2^TC_i^TD_iCe^{At_d}X_{12}(\Delta_{\hat{A}})^T\nonumber\\
      &\hspace*{0.5cm}+2t_dE_2^TC_i^TC_ie^{\hat{A}t_d}X_{22}(\Delta_{\hat{A}})^T-2t_dE_2^TC_i^TC_iE_1X_{12}(\Delta_{\hat{A}})^T\nonumber\\
      &\hspace*{0.5cm}-2t_dE_2^TC_i^TC_iE_2X_{22}(\Delta_{\hat{A}})^T-2t_de^{A_i^Tt_d}C_i^TD_iCe^{At_d}X_{13}(\Delta_{\hat{A}})^T\nonumber\\
      &\hspace*{0.5cm}-2t_de^{A_i^Tt_d}C_i^TC_iE_1X_{13}(\Delta_{\hat{A}})^T+2t_de^{\hat{A}^Tt_d}C_i^TC_ie^{A_it_d}X_{23}^T(\Delta_{\hat{A}})^T\nonumber\\
      &\hspace*{0.5cm}+2t_de^{A_i^Tt_d}C_i^TC_ie^{\hat{A}t_d}X_{23}(\Delta_{\hat{A}})^T-2t_de^{A_i^Tt_d}C_i^TC_iE_{2}X_{23}(\Delta_{\hat{A}})^T\nonumber\\
      &\hspace*{0.5cm}-2t_dE_2^TC_i^TC_ie^{A_it_d}X_{23}^T(\Delta_{\hat{A}})^T+2B_iCX_{13}\Delta_{Q_{33}}^{\hat{A}}-2B_i\hat{C}X_{23}\Delta_{Q_{33}}^{\hat{A}}\nonumber\\
      &\hspace*{0.5cm}-2C_i^TC_iE_1X_{11}(\Delta_{E_1}^{\hat{A}})^T-2C_i^TC_iE_2X_{12}^T(\Delta_{E_1}^{\hat{A}})^T\nonumber\\
      &\hspace*{0.5cm}-2C_i^TD_iCe^{At_d}X_{11}(\Delta_{E_1}^{\hat{A}})^T+2C_i^TC_ie^{\hat{A}t_d}X_{12}^T(\Delta_{E_1}^{\hat{A}})^T\Big).\nonumber
      \end{align}
      By applying Lemma \ref{lemma} on the equations (\ref{e45}) and (\ref{e28}), we get
      \begin{align}
      trace(R_6X_{33})=2trace(B_iCX_{13}\Delta_{Q_{33}}^{\hat{A}}-B_i\hat{C}X_{23}\Delta_{Q_{33}}^{\hat{A}}).\nonumber
      \end{align}
      Thus
      \begin{align}
      \Delta_{J}^{\hat{A}}&=trace\Big(2Q_{12}^TX_{12}(\Delta_{\hat{A}})^T+2Q_{22}X_{22}(\Delta_{\hat{A}})^T+2Q_{13}^TX_{13}(\Delta_{\hat{A}})^T\nonumber\\
      &+2Q_{23}X_{23}(\Delta_{\hat{A}})^T+2Q_{23}X_{23}^T(\Delta_{\hat{A}})^T+2Q_{33}X_{33}(\Delta_{\hat{A}})^T\nonumber\\
      &+2t_de^{\hat{A}^Tt_d}C_i^TD_iCe^{At_d}X_{12}(\Delta_{\hat{A}})^T-2t_de^{\hat{A}^Tt_d}C_i^TC_ie^{\hat{A}t_d}X_{22}(\Delta_{\hat{A}})^T\nonumber\\
      &+2t_de^{\hat{A}^Tt_d}C_i^TC_iE_1X_{12}(\Delta_{\hat{A}})^T+2t_de^{\hat{A}^Tt_d}C_i^TC_iE_2X_{22}(\Delta_{\hat{A}})^T\nonumber\\
      &-2t_dE_2^TC_i^TD_iCe^{At_d}X_{12}(\Delta_{\hat{A}})^T+2t_dE_2^TC_i^TC_ie^{\hat{A}t_d}X_{22}(\Delta_{\hat{A}})^T\nonumber\\
      &-2t_dE_2^TC_i^TC_iE_1X_{12}(\Delta_{\hat{A}})^T-2t_dE_2^TC_i^TC_iE_2X_{22}(\Delta_{\hat{A}})^T\nonumber\\
      &-2t_de^{A_i^Tt_d}C_i^TD_iCe^{At_d}X_{13}(\Delta_{\hat{A}})^T-2t_de^{A_i^Tt_d}C_i^TC_iE_1X_{13}(\Delta_{\hat{A}})^T\nonumber\\
      &+2t_de^{\hat{A}^Tt_d}C_i^TC_ie^{A_it_d}X_{23}^T(\Delta_{\hat{A}})^T+2t_de^{A_i^Tt_d}C_i^TC_ie^{\hat{A}t_d}X_{23}(\Delta_{\hat{A}})^T\nonumber\\
      &-2t_de^{A_i^Tt_d}C_i^TC_iE_{2}X_{23}(\Delta_{\hat{A}})^T-2t_dE_2^TC_i^TC_ie^{A_it_d}X_{23}^T(\Delta_{\hat{A}})^T\nonumber\\
      &-2t_de^{A_i^Tt_d}C_i^TC_ie^{A_it_d}X_{33}(\Delta_{\hat{A}})^T-2C_i^TC_iE_1X_{11}(\Delta_{E_1}^{\hat{A}})^T\nonumber\\
      &-2C_i^TC_iE_2X_{12}^T(\Delta_{E_1}^{\hat{A}})^T-2C_i^TD_iCe^{At_d}X_{11}(\Delta_{E_1}^{\hat{A}})^T\nonumber\\
      &+2C_i^TC_ie^{\hat{A}t_d}X_{12}^T(\Delta_{E_1}^{\hat{A}})^T\Big).\nonumber
      \end{align}
      By applying Lemma \ref{lemma} on the equations (\ref{e46}) and (\ref{e22}), we get
      \begin{align}
      trace(R_7^T\xi_1)&=trace(C_i^TC_iE_1X_{11}(\Delta_{E_1}^{\hat{A}})^T-C_i^TC_iE_2X_{12}^T(\Delta_{E_1}^{\hat{A}})^T\nonumber\\
      &\hspace*{1.5cm}-C_i^TD_iCe^{At_d}X_{11}(\Delta_{E_1}^{\hat{A}})^T+C_i^TC_ie^{\hat{A}t_d}X_{12}^T(\Delta_{E_1}^{\hat{A}})^T).\nonumber
      \end{align}
      Thus
      \begin{align}
      \Delta_{J}^{\hat{A}}&=trace\Big(2Q_{12}^TX_{12}(\Delta_{\hat{A}})^T+2Q_{22}X_{22}(\Delta_{\hat{A}})^T+2Q_{13}^TX_{13}(\Delta_{\hat{A}})^T\nonumber\\
      &+2Q_{23}X_{23}(\Delta_{\hat{A}})^T+2Q_{23}X_{23}^T(\Delta_{\hat{A}})^T+2Q_{33}X_{33}(\Delta_{\hat{A}})^T\nonumber\\
      &+2t_de^{\hat{A}^Tt_d}C_i^TD_iCe^{At_d}X_{12}(\Delta_{\hat{A}})^T-2t_de^{\hat{A}^Tt_d}C_i^TC_ie^{\hat{A}t_d}X_{22}(\Delta_{\hat{A}})^T\nonumber\\
      &+2t_de^{\hat{A}^Tt_d}C_i^TC_iE_1X_{12}(\Delta_{\hat{A}})^T+2t_de^{\hat{A}^Tt_d}C_i^TC_iE_2X_{22}(\Delta_{\hat{A}})^T\nonumber\\
      &-2t_dE_2^TC_i^TD_iCe^{At_d}X_{12}(\Delta_{\hat{A}})^T+2t_dE_2^TC_i^TC_ie^{\hat{A}t_d}X_{22}(\Delta_{\hat{A}})^T\nonumber\\
      &-2t_dE_2^TC_i^TC_iE_1X_{12}(\Delta_{\hat{A}})^T-2t_dE_2^TC_i^TC_iE_2X_{22}(\Delta_{\hat{A}})^T\nonumber\\
      &-2t_de^{A_i^Tt_d}C_i^TD_iCe^{At_d}X_{13}(\Delta_{\hat{A}})^T-2t_de^{A_i^Tt_d}C_i^TC_iE_1X_{13}(\Delta_{\hat{A}})^T\nonumber\\
      &+2t_de^{\hat{A}^Tt_d}C_i^TC_ie^{A_it_d}X_{23}^T(\Delta_{\hat{A}})^T+2t_de^{A_i^Tt_d}C_i^TC_ie^{\hat{A}t_d}X_{23}(\Delta_{\hat{A}})^T\nonumber\\
      &-2t_de^{A_i^Tt_d}C_i^TC_iE_{2}X_{23}(\Delta_{\hat{A}})^T-2t_dE_2^TC_i^TC_ie^{A_it_d}X_{23}^T(\Delta_{\hat{A}})^T\nonumber\\
      &-2t_de^{A_i^Tt_d}C_i^TC_ie^{A_it_d}X_{33}(\Delta_{\hat{A}})^T+2\xi_1E_1^T(\Delta_{\hat{A}})^T\nonumber\\
      &-2t_de^{A_i^Tt_d}\xi_1C^TB_i^T(\Delta_{\hat{A}})^T\Big).\nonumber
      \end{align}
      Since $e^{A_it_d}=e^{\hat{A}t_d}-E_2$, we get
      \begin{align}
      \Delta_{J}^{\hat{A}}&=trace\Big(2Q_{12}^TX_{12}(\Delta_{\hat{A}})^T+2Q_{22}X_{22}(\Delta_{\hat{A}})^T+2Q_{13}^TX_{13}(\Delta_{\hat{A}})^T\nonumber\\
      &+2Q_{23}X_{23}(\Delta_{\hat{A}})^T+2Q_{23}X_{23}^T(\Delta_{\hat{A}})^T+2Q_{33}X_{33}(\Delta_{\hat{A}})^T\nonumber\\
      &+2t_de^{A_i^Tt_d}C_i^TD_iCe^{At_d}X_{12}(\Delta_{\hat{A}})^T+2t_de^{A_i^Tt_d}C_i^TC_iE_1X_{12}(\Delta_{\hat{A}})^T\nonumber\\
      &-2t_de^{A_i^Tt_d}C_i^TD_iCe^{At_d}X_{13}(\Delta_{\hat{A}})^T-2t_de^{A_i^Tt_d}C_i^TC_iE_1X_{13}(\Delta_{\hat{A}})^T\nonumber\\
      &-2t_de^{A_i^Tt_d}C_i^TC_ie^{A_it_d}X_{22}(\Delta_{\hat{A}})^T+2t_de^{A_i^Tt_d}C_i^TC_ie^{A_it_d}X_{23}^T(\Delta_{\hat{A}})^T\nonumber\\
      &+2t_de^{A_i^Tt_d}C_i^TC_ie^{A_it_d}X_{23}(\Delta_{\hat{A}})^T-2t_de^{A_i^Tt_d}C_i^TC_ie^{A_it_d}X_{33}(\Delta_{\hat{A}})^T\nonumber\\
      &+2\xi_1E_1^T(\Delta_{\hat{A}})^T-2t_de^{A_i^Tt_d}\xi_1C^TB_i^T(\Delta_{\hat{A}})^T\Big).\nonumber
      \end{align}
      Hence,
      \begin{align}
      \frac{\partial}{\partial\hat{A}}||\Delta_{rel}(s)||_{\mathcal{H}_{2,\tau}}^2=2(Q_{12}^TX_{12}+Q_{22}X_{22}+\zeta_1),\nonumber
      \end{align} and
      \begin{align}
      Q_{12}^TX_{12}+Q_{22}X_{22}+\zeta_1=0\nonumber
      \end{align} is a necessary condition for the local optimum of $||\Delta_{rel}(s)||_{\mathcal{H}_{2,\tau}}^2$.

 Let us denote the first-order derivatives of $Q_{11}$, $Q_{12}$, $Q_{13}$, $Q_{22}$, $Q_{23}$, $Q_{33}$, $E_1$, and $J$ with respect to $\hat{B}$ as $\Delta_{Q_{11}}^{\hat{B}}$, $\Delta_{Q_{12}}^{\hat{B}}$, $\Delta_{Q_{13}}^{\hat{B}}$, $\Delta_{Q_{22}}^{\hat{B}}$, $\Delta_{Q_{23}}^{\hat{B}}$, $\Delta_{Q_{33}}^{\hat{B}}$, $\Delta_{E_1}^{\hat{B}}$, and $\Delta_{J}^{\hat{B}}$, respectively. Further, let us denote the differential of $\hat{B}$ as $\Delta_{\hat{B}}$. By differentiating $J$ with respect to $\hat{B}$, we get
      \begin{align}
      \Delta_{J}^{\hat{B}}&=trace(B^T\Delta_{Q_{11}}^{\hat{B}}B+2B^TQ_{12}\Delta_{\hat{B}}+2B^T\Delta_{Q_{12}}^{\hat{B}}\hat{B}+2\hat{B}^TQ_{22}\Delta_{\hat{B}}+\hat{B}^T\Delta_{Q_{22}}^{\hat{B}}\hat{B})\nonumber\\
      &=trace(2Q_{12}^TB(\Delta_{\hat{B}})^T+2Q_{22}\hat{B}(\Delta_{\hat{B}})^T+BB^T\Delta_{Q_{11}}^{\hat{B}}+2B\hat{B}^T(\Delta_{Q_{12}}^{\hat{B}})^T\nonumber\\
      &\hspace*{1.5cm}+\hat{B}\hat{B}^T\Delta_{Q_{22}}^{\hat{B}}).\nonumber
      \end{align}
      By differentiating the equations (\ref{e16}), (\ref{e17}), and (\ref{e19}) with respect to $\hat{B}$, we get
      \begin{align}
      A^T\Delta_{Q_{11}}^{\hat{B}}+\Delta_{Q_{11}}^{\hat{B}}A+S_1&=0,\label{e47}\\
      A^T\Delta_{Q_{12}}^{\hat{B}}+\Delta_{Q_{12}}^{\hat{B}}\hat{A}+S_2&=0,\label{e48}\\
      \hat{A}^T\Delta_{Q_{22}}^{\hat{B}}+\Delta_{Q_{22}}^{\hat{B}}\hat{A}+S_4&=0,\label{e49}
      \end{align}
      wherein
      \begin{align}
      S_1&=-C^TD_i^T(\Delta_{\hat{B}})^TQ_{13}^T-Q_{13}\Delta_{\hat{B}}D_iC+C^TB_i^T(\Delta_{Q_{13}}^{\hat{B}})^T+\Delta_{Q_{13}}^{\hat{B}}B_iC\nonumber\\
      &\hspace*{1cm}-e^{A^Tt_d}C^TD_i^TC_i\Delta_{E_1}^{\hat{B}}-(\Delta_{E_1}^{\hat{B}})^TC_i^TD_iCe^{At_d}-E_1^TC_i^TC_i\Delta_{E_1}^{\hat{B}}\nonumber\\
      &\hspace*{1cm}-(\Delta_{E_1}^{\hat{B}})^TC_i^TC_iE_1,\nonumber\\
S_2&=-C^TD_i^T(\Delta_{\hat{B}})^TQ_{23}+Q_{13}\Delta_{\hat{B}}C_i+C^TB_i^T\Delta_{Q_{23}}^{\hat{B}}-\Delta_{Q_{13}}^{\hat{B}}B_i\hat{C}\nonumber\\
&\hspace*{1cm}+(\Delta_{E_1}^{\hat{B}})^TC_i^TC_ie^{\hat{A}t_d}-t_de^{A^Tt_d}C^TD_i^TC_ie^{A_it_d}\Delta_{\hat{B}}C_i\nonumber\\
&\hspace*{1cm}-t_dE_1^TC_i^TC_ie^{A_it_d}\Delta_{\hat{B}}C_i-(\Delta_{E_1}^{\hat{B}})^TC_i^TC_iE_2,\nonumber\\
      S_4&=C_i^T(\Delta_{\hat{B}})^TQ_{23}+Q_{23}\Delta_{\hat{B}}C_i-\hat{C}^TB_i^T\Delta_{Q_{23}}^{\hat{B}}-\Delta_{Q_{23}}^{\hat{B}}B_i\hat{C}\nonumber\\
      &\hspace*{1cm}+t_dC_i^T(\Delta_{\hat{B}})^Te^{A_i^Tt_d}C_i^TC_ie^{\hat{A}t_d}+t_de^{\hat{A}^Tt_d}C_i^TC_ie^{A_it_d}\Delta_{\hat{B}}C_i\nonumber\\
      &\hspace*{1cm}-t_dC_i^T(\Delta_{\hat{B}})^Te^{A_i^Tt_d}C_i^TC_iE_2-t_dE_2^TC_i^TC_ie^{A_it_d}\Delta_{\hat{B}}C_i.\nonumber
      \end{align}
      By applying Lemma \ref{lemma} on the equations (\ref{e47}) and (\ref{e23}), (\ref{e48}) and (\ref{e24}), and (\ref{e49}) and (\ref{e26}), we get
      \begin{align}
      trace(BB^T\Delta_{Q_{11}}^{\hat{B}})&=trace(S_1X_{11}),\nonumber\\
      trace(B\hat{B}^T(\Delta_{Q_{12}}^{\hat{B}})^T)&=trace(S_2^TX_{12}),\nonumber\\
      trace(\hat{B}\hat{B}^T\Delta_{Q_{22}}^{\hat{B}})&=trace(S_4X_{22}).\nonumber
      \end{align}
Thus
\begin{align}
      \Delta_{J}^{\hat{B}}&=trace\Big(2Q_{12}^TB(\Delta_{\hat{B}})^T+2Q_{22}\hat{B}(\Delta_{\hat{B}})^T-2Q_{13}^TX_{11}C^TD_i^T(\Delta_{\hat{B}})^T\nonumber\\
      &\hspace*{1cm}+2Q_{23}X_{22}C_i^T(\Delta_{\hat{B}})^T-2Q_{23}X_{12}^TC^TD_i^T(\Delta_{\hat{B}})^T\nonumber\\
      &\hspace*{1cm}+2Q_{13}^TX_{12}C_i^T(\Delta_{\hat{B}})^T-2t_de^{A_i^Tt_d}C_i^TD_iCe^{At_d}X_{12}C_i^T(\Delta_{\hat{B}})^T\nonumber\\
      &\hspace*{1cm}-2t_de^{A_i^Tt_d}C_i^TC_iE_1X_{12}C_i^T(\Delta_{\hat{B}})^T+2t_de^{A_i^Tt_d}C_i^TC_ie^{\hat{A}t_d}X_{22}C_i^T(\Delta_{\hat{B}})^T\nonumber\\
      &\hspace*{1cm}-2t_de^{A_i^Tt_d}C_i^TC_iE_2X_{22}C_i^T(\Delta_{\hat{B}})^T+2X_{11}C^TB_i^T(\Delta_{Q_{13}}^{\hat{B}})^T\nonumber\\
      &\hspace*{1cm}-2X_{12}\hat{C}^TB_i^T(\Delta_{Q_{13}}^{\hat{B}})^T+2X_{12}^TC^TB_i^T\Delta_{Q_{23}}^{\hat{B}}-2X_{22}\hat{C}^TB_i^T\Delta_{Q_{23}}^{\hat{B}}\nonumber\\
      &\hspace*{1cm}-2C_i^TD_iCe^{At_d}X_{11}(\Delta_{E_1}^{\hat{B}})^T-2C_i^TC_iE_1X_{11}(\Delta_{E_1}^{\hat{B}})^T\nonumber\\
      &\hspace*{1cm}+2C_i^TC_ie^{\hat{A}t_d}X_{12}^T(\Delta_{E_1}^{\hat{B}})^T-2C_i^TC_iE_2X_{12}^T(\Delta_{E_1}^{\hat{B}})^T\Big).\nonumber
      \end{align}
      By differentiating the equations (\ref{e18}), (\ref{e20}), and (\ref{e21}) with respect to $\hat{B}$, we get
      \begin{align}
      A^T\Delta_{Q_{13}}^{\hat{B}}+\Delta_{Q_{13}}^{\hat{B}}A_i+S_3&=0,\label{e50}\\
      \hat{A}^T\Delta_{Q_{23}}^{\hat{B}}+\Delta_{Q_{23}}^{\hat{B}}A_i+S_5&=0,\label{e51}\\
      A_i^T\Delta_{Q_{33}}^{\hat{B}}+\Delta_{Q_{33}}^{\hat{B}}A_i+S_6&=0,\label{e52}
      \end{align}
      wherein
      \begin{align}
      S_3&=-Q_{13}\Delta_{\hat{B}}C_i-C^TD_i^T(\Delta_{\hat{B}})^TQ_{33}-C^TD_i^T\hat{B}^T\Delta_{Q_{33}}^{\hat{B}}\nonumber\\
      &+t_de^{A^Tt_d}C^TD_i^TC_ie^{A_it_d}\Delta_{\hat{B}}C_i-(\Delta_{E_1}^{\hat{B}})^TC_i^TC_ie^{A_it_d}+t_dE_1^TC_i^TC_ie^{A_it_d}\Delta_{\hat{B}}C_i,\nonumber\\
      S_5&=-Q_{23}\Delta_{\hat{B}}C_i+C_i^T(\Delta_{\hat{B}})^TQ_{33}-\hat{C}^TB_i^T\Delta_{Q_{33}}^{\hat{B}}\nonumber\\
      &-t_de^{\hat{A}^Tt_d}C_i^TC_ie^{A_it_d}\Delta_{\hat{B}}C_i+t_dE_2^TC_i^TC_ie^{A_it_d}\Delta_{\hat{B}}C_i-(\Delta_{E_2}^{\hat{B}})^TC_i^TC_ie^{A_it_d},\nonumber\\
      S_6&=-C_i^T(\Delta_{\hat{B}})^TQ_{33}-Q_{33}\Delta_{\hat{B}}C_i+t_dC_i^T(\Delta_{\hat{B}})^Te^{A_i^Tt_d}C_i^TC_ie^{A_it_d}\nonumber\\
      &+t_de^{A_i^Tt_d}C_i^TC_ie^{A_it_d}\Delta_{\hat{B}}C_i.\nonumber
      \end{align}
      By applying Lemma \ref{lemma} on the equations (\ref{e50}) and (\ref{e25}), and (\ref{e51}) and (\ref{e27}), we get
      \begin{align}
      trace(X_{11}C^TB_i^T(\Delta_{Q_{13}}^{\hat{B}})^T-X_{12}\hat{C}^TB_i^T(\Delta_{Q_{13}}^{\hat{B}})^T)=trace(S_{3}^TX_{13}),\nonumber\\
      trace(X_{12}^TC^TB_i^T\Delta_{Q_{23}}^{\hat{B}}-X_{22}\hat{C}^TB_i^T\Delta_{Q_{23}}^{\hat{B}})=trace(S_{5}^TX_{23}).\nonumber
      \end{align}
      Thus
      \begin{align}
      \Delta_{J}^{\hat{B}}&=trace\Big(2Q_{12}^TB(\Delta_{\hat{B}})^T+2Q_{22}\hat{B}(\Delta_{\hat{B}})^T-2Q_{13}^TX_{11}C^TD_i^T(\Delta_{\hat{B}})^T\nonumber\\
      &\hspace*{0.25cm}+2Q_{23}X_{22}C_i^T(\Delta_{\hat{B}})^T-2Q_{23}X_{12}^TC^TD_i^T(\Delta_{\hat{B}})^T\nonumber\\
      &\hspace*{0.25cm}+2Q_{13}^TX_{12}C_i^T(\Delta_{\hat{B}})^T-2Q_{13}^TX_{13}C_i^T(\Delta_{\hat{B}})^T\nonumber\\
      &\hspace*{0.25cm}-2Q_{23}X_{23}C_i^T(\Delta_{\hat{B}})^T+2Q_{33}X_{23}^TC_i^T(\Delta_{\hat{B}})^T\nonumber\\
      &\hspace*{0.25cm}-2Q_{33}X_{13}^TC^TD_i^T(\Delta_{\hat{B}})^T-2t_de^{A_i^Tt_d}C_i^TD_iCe^{At_d}X_{12}C_i^T(\Delta_{\hat{B}})^T\nonumber\\
      &\hspace*{0.25cm}-2t_de^{A_i^Tt_d}C_i^TC_iE_1X_{12}C_i^T(\Delta_{\hat{B}})^T+2t_de^{A_i^Tt_d}C_i^TC_ie^{\hat{A}t_d}X_{22}C_i^T(\Delta_{\hat{B}})^T\nonumber\\
      &\hspace*{0.25cm}-2t_de^{A_i^Tt_d}C_i^TC_iE_2X_{22}C_i^T(\Delta_{\hat{B}})^T+2t_de^{A_i^Tt_d}C_i^TD_iCe^{At_d}X_{13}C_i^T(\Delta_{\hat{B}})^T\nonumber\\
      &\hspace*{0.25cm}+2t_de^{A_i^Tt_d}C_i^TC_iE_1X_{13}C_i^T(\Delta_{\hat{B}})^T-2t_de^{A_i^Tt_d}C_i^TC_ie^{\hat{A}t_d}X_{23}C_i^T(\Delta_{\hat{B}})^T\nonumber\\
      &\hspace*{0.25cm}+2t_de^{A_i^Tt_d}C_i^TC_iE_2X_{23}C_i^T(\Delta_{\hat{B}})^T-2t_de^{A_i^Tt_d}C_i^TC_ie^{A_it_d}X_{23}^TC_i^T(\Delta_{\hat{B}})^T\nonumber\\
      &\hspace*{0.25cm}+2B_iCX_{13}\Delta_{Q_{33}}^{\hat{B}}-2B_i\hat{C}X_{23}\Delta_{Q_{33}}^{\hat{B}}-2C_i^TC_ie^{A_it_d}X_{13}^T(\Delta_{E_1}^{\hat{B}})^T\nonumber\\
      &\hspace*{0.25cm}-2C_i^TD_iCe^{At_d}X_{11}(\Delta_{E_1}^{\hat{B}})^T-2C_i^TC_iE_1X_{11}(\Delta_{E_1}^{\hat{B}})^T\nonumber\\
      &\hspace*{0.25cm}+2C_i^TC_ie^{\hat{A}t_d}X_{12}^T(\Delta_{E_1}^{\hat{B}})^T-2C_i^TC_iE_2X_{12}^T(\Delta_{E_1}^{\hat{B}})^T\Big).\nonumber
      \end{align}
      By applying Lemma \ref{lemma} on the equations (\ref{e52}) and (\ref{e28}), we get
      \begin{align}
      trace(S_6X_{33})=trace\big((2B_iCX_{13}-2B_i\hat{C}X_{23})\Delta_{Q_{33}}^{\hat{B}}\big).\nonumber
      \end{align}
      Thus
      \begin{align}
      \Delta_{J}^{\hat{B}}&=trace\Big(2Q_{12}^TB(\Delta_{\hat{B}})^T+2Q_{22}\hat{B}(\Delta_{\hat{B}})^T-2Q_{13}^TX_{11}C^TD_i^T(\Delta_{\hat{B}})^T\nonumber\\
      &+2Q_{23}X_{22}C_i^T(\Delta_{\hat{B}})^T-2Q_{23}X_{12}^TC^TD_i^T(\Delta_{\hat{B}})^T\nonumber\\
      &+2Q_{13}^TX_{12}C_i^T(\Delta_{\hat{B}})^T-2Q_{13}^TX_{13}C_i^T(\Delta_{\hat{B}})^T\nonumber\\
      &-2Q_{23}X_{23}C_i^T(\Delta_{\hat{B}})^T+2Q_{33}X_{23}^TC_i^T(\Delta_{\hat{B}})^T\nonumber\\
      &-2Q_{33}X_{13}^TC^TD_i^T(\Delta_{\hat{B}})^T-2Q_{33}X_{33}C_i^T(\Delta_{\hat{B}})^T\nonumber\\
      &-2t_de^{A_i^Tt_d}C_i^TD_iCe^{At_d}X_{12}C_i^T(\Delta_{\hat{B}})^T-2t_de^{A_i^Tt_d}C_i^TC_iE_1X_{12}C_i^T(\Delta_{\hat{B}})^T\nonumber\\
      &+2t_de^{A_i^Tt_d}C_i^TC_ie^{\hat{A}t_d}X_{22}C_i^T(\Delta_{\hat{B}})^T-2t_de^{A_i^Tt_d}C_i^TC_iE_2X_{22}C_i^T(\Delta_{\hat{B}})^T\nonumber\\
      &+2t_de^{A_i^Tt_d}C_i^TD_iCe^{At_d}X_{13}C_i^T(\Delta_{\hat{B}})^T+2t_de^{A_i^Tt_d}C_i^TC_iE_1X_{13}C_i^T(\Delta_{\hat{B}})^T\nonumber\\
      &-2t_de^{A_i^Tt_d}C_i^TC_ie^{\hat{A}t_d}X_{23}C_i^T(\Delta_{\hat{B}})^T+2t_de^{A_i^Tt_d}C_i^TC_iE_2X_{23}C_i^T(\Delta_{\hat{B}})^T\nonumber\\
      &-2t_de^{A_i^Tt_d}C_i^TC_ie^{A_it_d}X_{23}^TC_i^T(\Delta_{\hat{B}})^T+2t_de^{A_i^Tt_d}C_i^TC_ie^{A_it_d}X_{33}C_i^T(\Delta_{\hat{B}})^T\nonumber\\
      &-2C_i^TC_ie^{A_it_d}X_{13}^T(\Delta_{E_1}^{\hat{B}})^T-2C_i^TD_iCe^{At_d}X_{11}(\Delta_{E_1}^{\hat{B}})^T\nonumber\\
      &-2C_i^TC_iE_1X_{11}(\Delta_{E_1}^{\hat{B}})^T+2C_i^TC_ie^{\hat{A}t_d}X_{12}^T(\Delta_{E_1}^{\hat{B}})^T-2C_i^TC_iE_2X_{12}^T(\Delta_{E_1}^{\hat{B}})^T\Big).\nonumber
      \end{align}
By differentiating the equation (\ref{e8}) with respect to $\hat{B}$, we get
      \begin{align}
      A_i\Delta_{E_1}^{\hat{B}}-\Delta_{E_1}^{\hat{B}}A+S_7=0\label{e53}
      \end{align}where
      \begin{align}
      S_7=-\Delta_{\hat{B}}C_iE_1-\Delta_{\hat{B}}D_iCe^{At_d}+e^{A_it_d}\Delta_{\hat{B}}D_iC+t_de^{A_it_d}\Delta_{\hat{B}}C_iB_iC.\nonumber
      \end{align}
       By applying Lemma \ref{lemma} on the equations (\ref{e53}) and (\ref{e29}), we get
      \begin{align}
      trace(S_7^T\xi_2)&=trace\Big(\big(-C_i^TC_ie^{A_it_d}X_{13}^T-C_i^TD_iCe^{At_d}X_{11}-C_i^TC_iE_1X_{11}\nonumber\\
      &\hspace*{2cm}+C_i^TC_ie^{\hat{A}t_d}X_{12}^T-2C_i^TC_iE_2X_{12}^T\big)(\Delta_{E_1}^{\hat{B}})^T\big).\nonumber
      \end{align}
      Thus
      \begin{align}
      \Delta_{J}^{\hat{B}}&=trace\Big(2Q_{12}^TB(\Delta_{\hat{B}})^T+2Q_{22}\hat{B}(\Delta_{\hat{B}})^T-2Q_{13}^TX_{11}C^TD_i^T(\Delta_{\hat{B}})^T\nonumber\\
      &+2Q_{23}X_{22}C_i^T(\Delta_{\hat{B}})^T-2Q_{23}X_{12}^TC^TD_i^T(\Delta_{\hat{B}})^T\nonumber\\
      &+2Q_{13}^TX_{12}C_i^T(\Delta_{\hat{B}})^T-2Q_{13}^TX_{13}C_i^T(\Delta_{\hat{B}})^T\nonumber\\
      &-2Q_{23}X_{23}C_i^T(\Delta_{\hat{B}})^T+2Q_{33}X_{23}^TC_i^T(\Delta_{\hat{B}})^T\nonumber\\
      &-2Q_{33}X_{13}^TC^TD_i^T(\Delta_{\hat{B}})^T-2Q_{33}X_{33}C_i^T(\Delta_{\hat{B}})^T\nonumber\\
      &-2t_de^{A_i^Tt_d}C_i^TD_iCe^{At_d}X_{12}C_i^T(\Delta_{\hat{B}})^T-2t_de^{A_i^Tt_d}C_i^TC_iE_1X_{12}C_i^T(\Delta_{\hat{B}})^T\nonumber\\
      &+2t_de^{A_i^Tt_d}C_i^TC_ie^{\hat{A}t_d}X_{22}C_i^T(\Delta_{\hat{B}})^T-2t_de^{A_i^Tt_d}C_i^TC_iE_2X_{22}C_i^T(\Delta_{\hat{B}})^T\nonumber\\
      &+2t_de^{A_i^Tt_d}C_i^TD_iCe^{At_d}X_{13}C_i^T(\Delta_{\hat{B}})^T+2t_de^{A_i^Tt_d}C_i^TC_iE_1X_{13}C_i^T(\Delta_{\hat{B}})^T\nonumber\\
      &-2t_de^{A_i^Tt_d}C_i^TC_ie^{\hat{A}t_d}X_{23}C_i^T(\Delta_{\hat{B}})^T+2t_de^{A_i^Tt_d}C_i^TC_iE_2X_{23}C_i^T(\Delta_{\hat{B}})^T\nonumber\\
      &-2t_de^{A_i^Tt_d}C_i^TC_ie^{A_it_d}X_{23}^TC_i^T(\Delta_{\hat{B}})^T+2t_de^{A_i^Tt_d}C_i^TC_ie^{A_it_d}X_{33}C_i^T(\Delta_{\hat{B}})^T\nonumber\\
      &-2\xi_2E_1^TC_i^T(\Delta_{\hat{B}})^T-2\xi_2e^{A^Tt_d}C^TD_i(\Delta_{\hat{B}})^T+2e^{A_i^Tt_d}\xi_2C^TD_i(\Delta_{\hat{B}})^T\nonumber\\
      &+2t_de^{A_i^Tt_d}\xi_2C^TB_i^TC_i^T(\Delta_{\hat{B}})^T\Big).\nonumber
      \end{align}
       Since $e^{A_it_d}=e^{\hat{A}t_d}-E_2$, we get
       \begin{align}
      \Delta_{J}^{\hat{B}}&=trace\Big(2Q_{12}^TB(\Delta_{\hat{B}})^T+2Q_{22}\hat{B}(\Delta_{\hat{B}})^T-2Q_{13}^TX_{11}C^TD_i^T(\Delta_{\hat{B}})^T\nonumber\\
      &+2Q_{23}X_{22}C_i^T(\Delta_{\hat{B}})^T-2Q_{23}X_{12}^TC^TD_i^T(\Delta_{\hat{B}})^T\nonumber\\
      &+2Q_{13}^TX_{12}C_i^T(\Delta_{\hat{B}})^T-2Q_{33}X_{13}^TC^TD_i^T(\Delta_{\hat{B}})^T\nonumber\\
      &-2Q_{13}^TX_{13}C_i^T(\Delta_{\hat{B}})^T-2Q_{23}X_{23}C_i^T(\Delta_{\hat{B}})^T\nonumber\\
      &+2Q_{33}X_{23}^TC_i^T(\Delta_{\hat{B}})^T-2Q_{33}X_{33}C_i^T(\Delta_{\hat{B}})^T\nonumber\\
      &-2t_de^{A_i^Tt_d}C_i^TD_iCe^{At_d}X_{12}C_i^T(\Delta_{\hat{B}})^T-2t_de^{A_i^Tt_d}C_i^TC_iE_1X_{12}C_i^T(\Delta_{\hat{B}})^T\nonumber\\
      &+2t_de^{A_i^Tt_d}C_i^TD_iCe^{At_d}X_{13}C_i^T(\Delta_{\hat{B}})^T+2t_de^{A_i^Tt_d}C_i^TC_iE_1X_{13}C_i^T(\Delta_{\hat{B}})^T\nonumber\\
      &+2t_de^{A_i^Tt_d}C_i^TC_ie^{A_it_d}X_{22}C_i^T(\Delta_{\hat{B}})^T-2t_de^{A_i^Tt_d}C_i^TC_ie^{A_it_d}X_{23}C_i^T(\Delta_{\hat{B}})^T\nonumber\\
      &-2t_de^{A_i^Tt_d}C_i^TC_ie^{A_it_d}X_{23}^TC_i^T(\Delta_{\hat{B}})^T+2t_de^{A_i^Tt_d}C_i^TC_ie^{A_it_d}X_{33}C_i^T(\Delta_{\hat{B}})^T\nonumber\\
      &-2\xi_2E_1^TC_i^T(\Delta_{\hat{B}})^T-2\xi_2e^{A^Tt_d}C^TD_i^T(\Delta_{\hat{B}})^T+2e^{A_i^Tt_d}\xi_2C^TD_i^T(\Delta_{\hat{B}})^T\nonumber\\
      &+2t_de^{A_i^Tt_d}\xi_2C^TB_i^TC_i^T(\Delta_{\hat{B}})^T\Big).\nonumber
      \end{align}
      Hence,
      \begin{align}
      \frac{\partial}{\partial\hat{B}}||\Delta_{rel}(s)||_{\mathcal{H}_{2,\tau}}^2=2(Q_{12}^TB+Q_{22}\hat{B}+\zeta_2),\nonumber
      \end{align} and
      \begin{align}
      Q_{12}^TB+Q_{22}\hat{B}+\zeta_2=0\nonumber
      \end{align} is a necessary condition for the local optimum of $||\Delta_{rel}(s)||_{\mathcal{H}_{2,\tau}}^2$.

Let us denote the first-order derivatives of $P_{13}$, $P_{23}$, $P_{33}$, $E_1$, and $J$ with respect to $\hat{C}$ as $\Delta_{P_{13}}^{\hat{C}}$, $\Delta_{P_{23}}^{\hat{C}}$, $\Delta_{P_{33}}^{\hat{C}}$, $\Delta_{E_1}^{\hat{C}}$, and $\Delta_{J}^{\hat{C}}$, respectively. Further, let us denote the differential of $\hat{C}$ as $\Delta_{\hat{C}}$. By taking differentiation of $J$ with respect to $\hat{C}$, we get
          \begin{align}
      \Delta_{J}^{\hat{C}}&=trace\Big(-2D_iCP_{12}(\Delta_{\hat{C}})^TD_i^T+2D_iCP_{13}(\Delta_{\hat{C}})^TD_i^T+2D_i\hat{C}P_{22}(\Delta_{\hat{C}})^TD_i^T\nonumber\\
      &\hspace*{2cm}-2D_i\hat{C}P_{23}(\Delta_{\hat{C}})^TD_i^T-2C_iP_{23}^T(\Delta_{\hat{C}})^TD_i^T+2C_iP_{33}(\Delta_{\hat{C}})^TD_i^T\nonumber\\
      &\hspace*{2cm}+2D_iC\Delta_{P_{13}}^{\hat{C}}C_i^T-2D_i\hat{C}\Delta_{P_{23}}^{\hat{C}}C_i^T+C_i\Delta_{P_{33}}^{\hat{C}}C_i^T\Big)\nonumber\\
      &=trace\Big(-2D_i^TD_iCP_{12}(\Delta_{\hat{C}})^T+2D_i^TD_iCP_{13}(\Delta_{\hat{C}})^T+2D_i^TD_i\hat{C}P_{22}(\Delta_{\hat{C}})^T\nonumber\\
      &\hspace*{2cm}-2D_i^TD_i\hat{C}P_{23}(\Delta_{\hat{C}})^T-2D_i^TC_iP_{23}^T(\Delta_{\hat{C}})^T+2D_i^TC_iP_{33}(\Delta_{\hat{C}})^T\nonumber\\
      &\hspace*{2cm}+2C^TD_i^TC_i(\Delta_{P_{13}}^{\hat{C}})^T-2\hat{C}^TD_i^TC_i(\Delta_{P_{23}}^{\hat{C}})^T+C_i^TC_i\Delta_{P_{33}}^{\hat{C}}\Big).\nonumber
      \end{align}
By taking differentiation of the equations (\ref{e11}), (\ref{e13}), (\ref{e14}), and (\ref{e8}) with respect to $\hat{C}$, we get
\begin{align}
A\Delta_{P_{13}}^{\hat{C}}+\Delta_{P_{13}}^{\hat{C}}A_i^T+T_1&=0,\label{e54}\\
\hat{A}\Delta_{P_{23}}^{\hat{C}}+\Delta_{P_{23}}^{\hat{C}}A_i^T+T_2&=0,\label{e55}\\
A_i\Delta_{P_{33}}^{\hat{C}}+\Delta_{P_{33}}^{\hat{C}}A_i^T+T_3&=0,\label{e56}\\
A_i\Delta_{E_1}^{\hat{C}}-\Delta_{E_1}^{\hat{C}}A+T_4&=0,\label{e57}
\end{align}wherein
\begin{align}
T_1&=P_{13}(\Delta_{\hat{C}})^TB_i^T-P_{12}(\Delta_{\hat{C}})^TB_i^T-e^{At_d}BB^T(\Delta_{E_1}^{\hat{C}})^T\nonumber\\
&+t_de^{At_d}B\hat{B}^T(\Delta_{\hat{C}})^TB_i^Te^{A_i^Tt_d},\nonumber\\
T_2&=P_{23}(\Delta_{\hat{C}})^TB_i^T-P_{22}(\Delta_{\hat{C}})^TB_i^T-e^{\hat{A}t_d}\hat{B}B^T(\Delta_{E_1}^{\hat{C}})^T\nonumber\\
&+t_de^{\hat{A}t_d}\hat{B}\hat{B}^T(\Delta_{\hat{C}})^TB_i^Te^{A_i^Tt_d},\nonumber\\
T_3&=B_i\Delta_{\hat{C}}P_{33}+P_{33}(\Delta_{\hat{C}})^TB_i^T+B_iC\Delta_{P_{13}}^{\hat{C}}-B_i\hat{C}\Delta_{P_{23}}^{\hat{C}}-B_i\Delta_{\hat{C}}P_{23}\nonumber\\
&+(\Delta_{P_{13}}^{\hat{C}})^TC^TB_i^T-(\Delta_{P_{23}}^{\hat{C}})^T\hat{C}^TB_i^T-P_{23}^T(\Delta_{\hat{C}})^TB_i^T-E_1BB^T(\Delta_{E_1}^{\hat{C}})^T\nonumber\\
&-\Delta_{E_1}^{\hat{C}}BB^TE_1^T+t_de^{A_it_d}B_i\Delta_{\hat{C}}\hat{B}B^TE_1^T-E_2\hat{B}B^T(\Delta_{E_1}^{\hat{C}})^T-\Delta_{E_1}^{\hat{C}}B\hat{B}^TE_2^T\nonumber\\
&+t_dE_1B\hat{B}^T(\Delta_{\hat{C}})^TB_i^Te^{A_i^Tt_d}+t_dE_2\hat{B}\hat{B}^T(\Delta_{\hat{C}})^TB_i^Te^{A_i^Tt_d}\nonumber\\
&+t_de^{A_it_d}B_i\Delta_{\hat{C}}\hat{B}\hat{B}^TE_2^T,\nonumber\\
T_4&=B_i\Delta_{\hat{C}}E_1-t_de^{A_it_d}B_i\Delta_{\hat{C}}B_i\hat{C}.\nonumber
\end{align}
By applying Lemma \ref{lemma} on the equations (\ref{e54}) and (\ref{e30}), (\ref{e55}) and (\ref{e31}), and (\ref{e56}) and (\ref{e32}), we get
\begin{align}
trace(T_1^TY_{13})&=trace(C^TB_i^TY_{33}(\Delta_{P_{13}}^{\hat{C}})^T+C^TD_i^TC_i(\Delta_{P_{13}}^{\hat{C}})^T)\nonumber\\
trace(T_2^TY_{23})&=trace(-\hat{C}^TB_i^TY_{33}(\Delta_{P_{23}}^{\hat{C}})^T-\hat{C}^TD_i^TC_i(\Delta_{P_{23}}^{\hat{C}})^T)\nonumber\\
trace(T_3Y_{33})&=trace(C_i^TC_i\Delta_{P_{33}}^{\hat{C}})\nonumber.
\end{align}
Thus
\begin{align}
      \Delta_{J}^{\hat{C}}&=trace\Big(-2D_i^TD_iCP_{12}(\Delta_{\hat{C}})^T+2D_i^TD_iCP_{13}(\Delta_{\hat{C}})^T+2D_i^TD_i\hat{C}P_{22}(\Delta_{\hat{C}})^T\nonumber\\
      &-2D_i^TD_i\hat{C}P_{23}(\Delta_{\hat{C}})^T-2D_i^TC_iP_{23}^T(\Delta_{\hat{C}})^T+2D_i^TC_iP_{33}(\Delta_{\hat{C}})^T\nonumber\\
      &+2B_i^TY_{13}^TP_{13}(\Delta_{\hat{C}})^T+2B_i^TY_{13}^TP_{12}(\Delta_{\hat{C}})^T+2B_i^TY_{23}P_{23}(\Delta_{\hat{C}})^T\nonumber\\
      &-2B_i^TY_{23}P_{22}(\Delta_{\hat{C}})^T-2B_i^TY_{33}P_{23}^T(\Delta_{\hat{C}})^T+2t_dB_i^Te^{A_i^Tt_d}Y_{13}^Te^{At_d}B\hat{B}^T(\Delta_{\hat{C}})^T\nonumber\\
      &+2t_dB_i^Te^{A_i^Tt_d}Y_{23}e^{\hat{A}t_d}\hat{B}\hat{B}^T(\Delta_{\hat{C}})^T+2t_dB_i^Te^{A_i^Tt_d}Y_{33}E_2\hat{B}\hat{B}^T(\Delta_{\hat{C}})^T\nonumber\\
      &+2t_dB_i^Te^{A_i^Tt_d}Y_{33}E_1B\hat{B}^T(\Delta_{\hat{C}})^T-2Y_{13}^Te^{At_d}BB^T(\Delta_{E_1}^{\hat{C}})^T\nonumber\\
      &-2Y_{23}e^{\hat{A}t_d}\hat{B}B^T(\Delta_{E_1}^{\hat{C}})^T-2Y_{33}E_1BB^T(\Delta_{E_1}^{\hat{C}})^T-2Y_{33}E_2\hat{B}B^T(\Delta_{E_1}^{\hat{C}})^T\Big).\nonumber
      \end{align}
By applying Lemma \ref{lemma} on the equations (\ref{e57}) and (\ref{e33}), we get
\begin{align}
trace(T_4^T\xi_3)&=trace\big((Y_{13}^Te^{At_d}BB^T-Y_{23}e^{\hat{A}t_d}\hat{B}B^T-Y_{33}E_1BB^T\nonumber\\
&\hspace*{5cm}-Y_{33}E_2\hat{B}B^T)(\Delta_{E_1}^{\hat{C}})^T\big).\nonumber
\end{align}
Thus
\begin{align}
      \Delta_{J}^{\hat{C}}&=trace\Big(-2D_i^TD_iCP_{12}(\Delta_{\hat{C}})^T+2D_i^TD_iCP_{13}(\Delta_{\hat{C}})^T+2D_i^TD_i\hat{C}P_{22}(\Delta_{\hat{C}})^T\nonumber\\
      &-2D_i^TD_i\hat{C}P_{23}(\Delta_{\hat{C}})^T-2D_i^TC_iP_{23}^T(\Delta_{\hat{C}})^T+2D_i^TC_iP_{33}(\Delta_{\hat{C}})^T\nonumber\\
      &+2B_i^TY_{13}^TP_{13}(\Delta_{\hat{C}})^T+2B_i^TY_{13}^TP_{12}(\Delta_{\hat{C}})^T+2B_i^TY_{23}P_{23}(\Delta_{\hat{C}})^T\nonumber\\
      &-2B_i^TY_{23}P_{22}(\Delta_{\hat{C}})^T-2B_i^TY_{33}P_{23}^T(\Delta_{\hat{C}})^T+2t_dB_i^Te^{A_i^Tt_d}Y_{13}^Te^{At_d}B\hat{B}^T(\Delta_{\hat{C}})^T\nonumber\\
      &+2t_dB_i^Te^{A_i^Tt_d}Y_{23}e^{\hat{A}t_d}\hat{B}\hat{B}^T(\Delta_{\hat{C}})^T+2t_dB_i^Te^{A_i^Tt_d}Y_{33}E_2\hat{B}\hat{B}^T(\Delta_{\hat{C}})^T\nonumber\\
      &+2t_dB_i^Te^{A_i^Tt_d}Y_{33}E_1B\hat{B}^T(\Delta_{\hat{C}})^T-2t_dB_i^Te^{A_i^Tt_d}\xi_3\hat{C}^TB_i^T(\Delta_{\hat{C}})^T\nonumber\\
      &+2B_i^T\xi_3E_1^T(\Delta_{\hat{C}})^T\Big).\nonumber
      \end{align}
Hence,
      \begin{align}
      \frac{\partial}{\partial\hat{C}}||\Delta_{rel}(s)||_{\mathcal{H}_{2,\tau}}^2=2(-D_i^TD_iCP_{12}+D_i^TD_i\hat{C}P_{22}+\zeta_3),\nonumber
      \end{align} and
      \begin{align}
      -D_i^TD_iCP_{12}+D_i^TD_i\hat{C}P_{22}+\zeta_3=0\nonumber
      \end{align} is a necessary condition for the local optimum of $||\Delta_{rel}(s)||_{\mathcal{H}_{2,\tau}}^2$. This completes the proof.
\section*{Acknowledgment}
This work is supported by the National Natural Science Foundation of China under Grants No. 61873336 and 61803046, the International Corporation Project of Shanghai Science and Technology Commission under Grant 21190780300, in part by The National Key Research and Development Program No. 2020YFB1708200 , and in part by the High-end foreign expert program No. G2021013008L granted by the State Administration of Foreign Experts Affairs (SAFEA).
\section*{Competing Interest Statement}
The authors declare no conflict of interest.
%\bibliography{flbib}

\begin{thebibliography}{}
\expandafter\ifx\csname url\endcsname\relax
  \def\url#1{\texttt{#1}}\fi
\expandafter\ifx\csname urlprefix\endcsname\relax\def\urlprefix{URL }\fi
\expandafter\ifx\csname href\endcsname\relax
  \def\href#1#2{#2} \def\path#1{#1}\fi

\end{thebibliography}


\begin{thebibliography}{10}
\expandafter\ifx\csname url\endcsname\relax
  \def\url#1{\texttt{#1}}\fi
\expandafter\ifx\csname urlprefix\endcsname\relax\def\urlprefix{URL }\fi
\expandafter\ifx\csname href\endcsname\relax
  \def\href#1#2{#2} \def\path#1{#1}\fi

\bibitem{schilders2008model}
W.~H. Schilders, H.~A. Van~der Vorst, J.~Rommes, Model order reduction: theory,
  research aspects and applications, Vol.~13, Springer, 2008.

\bibitem{benner2011model}
P.~Benner, M.~Hinze, E.~J.~W. Ter~Maten, Model reduction for circuit
  simulation, Vol.~74, Springer, 2011.

\bibitem{benner2005dimension}
P.~Benner, V.~Mehrmann, D.~C. Sorensen, Dimension reduction of large-scale
  systems, Vol.~45, Springer, 2005.

\bibitem{obinata2012model}
G.~Obinata, B.~D. Anderson, Model reduction for control system design, Springer
  Science \& Business Media, 2012.

\bibitem{benner2020model}
P.~Benner, W.~Schilders, S.~Grivet-Talocia, A.~Quarteroni, G.~Rozza,
  L.~Miguel~Silveira, Model order reduction: Volume 3 applications, De Gruyter,
  2020.

\bibitem{moore1981principal}
B.~Moore, Principal component analysis in linear systems: Controllability,
  observability, and model reduction, IEEE transactions on automatic control
  26~(1) (1981) 17--32.

\bibitem{enns1984model}
D.~F. Enns, Model reduction with balanced realizations: An error bound and a
  frequency weighted generalization, in: The 23rd IEEE conference on decision
  and control, IEEE, 1984, pp. 127--132.

\bibitem{gugercin2005smith}
S.~Gugercin, J.-R. Li, Smith-type methods for balanced truncation of large
  sparse systems, in: Dimension reduction of large-scale systems, Springer,
  2005, pp. 49--82.

\bibitem{li2002low}
J.-R. Li, J.~White, Low rank solution of Lyapunov equations, SIAM Journal on
  Matrix Analysis and Applications 24~(1) (2002) 260--280.

\bibitem{penzl1999cyclic}
T.~Penzl, A cyclic low-rank Smith method for large sparse Lyapunov equations,
  SIAM Journal on Scientific Computing 21~(4) (1999) 1401--1418.

\bibitem{kurschner2020inexact}
P.~K{\"u}rschner, M.~A. Freitag, Inexact methods for the low rank solution to
  large scale Lyapunov equations, BIT Numerical Mathematics 60~(4) (2020)
  1221--1259.

\bibitem{mehrmann2005balanced}
V.~Mehrmann, T.~Stykel, Balanced truncation model reduction for large-scale
  systems in descriptor form, in: Dimension Reduction of Large-Scale Systems,
  Springer, 2005, pp. 83--115.

\bibitem{tombs1987truncated}
M.~S. Tombs, I.~Postlethwaite, Truncated balanced realization of a stable
  non-minimal state-space system, International Journal of Control 46~(4)
  (1987) 1319--1330.

\bibitem{safonov1989schur}
M.~G. Safonov, R.~Chiang, A schur method for balanced-truncation model
  reduction, IEEE Transactions on Automatic Control 34~(7) (1989) 729--733.

\bibitem{phillips2003guaranteed}
J.~R. Phillips, L.~Daniel, L.~M. Silveira, Guaranteed passive balancing
  transformations for model order reduction, IEEE Transactions on
  Computer-Aided Design of Integrated Circuits and Systems 22~(8) (2003)
  1027--1041.

\bibitem{yan2008second}
B.~Yan, S.~X.-D. Tan, B.~McGaughy, Second-order balanced truncation for
  passive-order reduction of RLCK circuits, IEEE Transactions on Circuits and
  Systems II: Express Briefs 55~(9) (2008) 942--946.

\bibitem{phillips2004poor}
J.~R. Phillips, L.~M. Silveira, Poor man's TBR: A simple model reduction
  scheme, IEEE transactions on computer-aided design of integrated circuits and
  systems 24~(1) (2004) 43--55.

\bibitem{reis2010positive}
T.~Reis, T.~Stykel, Positive real and bounded real balancing for model
  reduction of descriptor systems, International Journal of Control 83~(1)
  (2010) 74--88.

\bibitem{reis2008balanced}
T.~Reis, T.~Stykel, Balanced truncation model reduction of second-order
  systems, Mathematical and Computer Modelling of Dynamical Systems 14~(5)
  (2008) 391--406.

\bibitem{wang2011balanced}
X.~Wang, Q.~Wang, Z.~Zhang, Q.~Chen, N.~Wong, Balanced truncation for
  time-delay systems via approximate gramians, in: 16th Asia and South Pacific
  Design Automation Conference (ASP-DAC 2011), IEEE, 2011, pp. 55--60.

\bibitem{zhang2003gramians}
L.~Zhang, J.~Lam, B.~Huang, G.-H. Yang, On gramians and balanced truncation of
  discrete-time bilinear systems, International Journal of Control 76~(4)
  (2003) 414--427.

\bibitem{lall2002subspace}
S.~Lall, J.~E. Marsden, S.~Glava{\v{s}}ki, A subspace approach to balanced
  truncation for model reduction of nonlinear control systems, International
  Journal of Robust and Nonlinear Control: IFAC-Affiliated Journal 12~(6)
  (2002) 519--535.

\bibitem{glover1984all}
K.~Glover, All optimal Hankel-norm approximations of linear multivariable
  systems and their $L_\infty$-error bounds, International journal of control
  39~(6) (1984) 1115--1193.

\bibitem{glover1989tutorial}
K.~Glover, A tutorial on Hankel-norm approximation, From data to model (1989)
  26--48.

\bibitem{green1988balanced}
M.~Green, Balanced stochastic realizations, Linear Algebra and its Applications
  98 (1988) 211--247.

\bibitem{green1988relative}
M.~Green, A relative error bound for balanced stochastic truncation, IEEE
  Transactions on Automatic Control 33~(10) (1988) 961--965.

\bibitem{flagg2010interpolation}
G.~M. Flagg, S.~Gugercin, C.~A. Beattie, An interpolation-based approach to
  $\mathcal{H}_\infty$ model reduction of dynamical systems, in: 49th IEEE
  Conference on Decision and Control (CDC), IEEE, 2010, pp. 6791--6796.

\bibitem{flagg2013interpolatory}
G.~Flagg, C.~A. Beattie, S.~Gugercin, Interpolatory $\mathcal{H}_\infty$ model
  reduction, Systems \& Control Letters 62~(7) (2013) 567--574.

\bibitem{castagnotto2017interpolatory}
A.~Castagnotto, C.~Beattie, S.~Gugercin, Interpolatory methods for
  $\mathcal{H}_\infty$ model reduction of multi-input/multi-output systems, in:
  Model Reduction of Parametrized Systems, Springer, 2017, pp. 349--365.

\bibitem{wolf2014h}
T.~Wolf, $\mathcal{H}_2$ pseudo-optimal model order reduction, Ph.D. thesis,
  Technische Universit{\"a}t M{\"u}nchen (2014).

\bibitem{wilson1970optimum}
D.~Wilson, Optimum solution of model-reduction problem, in: Proceedings of the
  Institution of Electrical Engineers, Vol. 117, IET, 1970, pp. 1161--1165.

\bibitem{gugercin2008h_2}
S.~Gugercin, A.~C. Antoulas, C.~Beattie, $\mathcal{H}_2$ model reduction for
  large-scale linear dynamical systems, SIAM journal on matrix analysis and
  applications 30~(2) (2008) 609--638.

\bibitem{van2008h2}
P.~Van~Dooren, K.~A. Gallivan, P.-A. Absil, $\mathcal{H}_2$-optimal model
  reduction of MIMO systems, Applied Mathematics Letters 21~(12) (2008)
  1267--1273.

\bibitem{xu2011optimal}
Y.~Xu, T.~Zeng, Optimal $\mathcal{H}_2$ model reduction for large scale MIMO
  systems via tangential interpolation., International Journal of Numerical
  Analysis \& Modeling 8~(1) (2011).

\bibitem{gallivan2004sylvester}
K.~Gallivan, A.~Vandendorpe, P.~Van~Dooren, Sylvester equations and
  projection-based model reduction, Journal of Computational and Applied
  Mathematics 162~(1) (2004) 213--229.

\bibitem{benner2011sparse}
P.~Benner, M.~K{\o}hler, J.~Saak, Sparse-dense Sylvester equations in
  $\mathcal{H}_2$-model order reduction., MPI. Magdeburg preprints MPIMD/11-11 (2011).

\bibitem{beattie2009trust}
C.~A. Beattie, S.~Gugercin, A trust region method for optimal $\mathcal{H}_2$
  model reduction, in: Proceedings of the 48h IEEE Conference on Decision and
  Control (CDC) held jointly with 2009 28th Chinese Control Conference, IEEE,
  2009, pp. 5370--5375.

\bibitem{jiang2017h2}
Y.~Jiang, K.~Xu, $\mathcal{H}_2$ optimal reduced models of general MIMO LTI
  systems via the cross gramian on the Stiefel manifold, Journal of the
  Franklin Institute 354~(8) (2017) 3210--3224.

\bibitem{sato2017structure}
K.~Sato, H.~Sato, Structure-preserving $\mathcal{H}_2$ optimal model reduction
  based on the Riemannian trust-region method, IEEE Transactions on Automatic
  Control 63~(2) (2017) 505--512.

\bibitem{sato2015riemannian}
H.~Sato, K.~Sato, Riemannian trust-region methods for $\mathcal{H}_2$ optimal
  model reduction, in: 2015 54th IEEE Conference on Decision and Control (CDC),
  IEEE, 2015, pp. 4648--4655.

\bibitem{yang2019trust}
P.~Yang, Y.-L. Jiang, K.-L. Xu, A trust-region method for $\mathcal{H}_2$ model
  reduction of bilinear systems on the Stiefel manifold, Journal of the
  Franklin Institute 356~(4) (2019) 2258--2273.

\bibitem{panzer2013greedy}
H.~K. Panzer, S.~Jaensch, T.~Wolf, B.~Lohmann, A greedy rational Krylov method
  for $\mathcal{H}_2$-pseudooptimal model order reduction with preservation of
  stability, in: 2013 American Control Conference, IEEE, 2013, pp. 5512--5517.

\bibitem{wolf2013h}
T.~Wolf, H.~K. Panzer, B.~Lohmann, $\mathcal{H}_2$ pseudo-optimality in model
  order reduction by Krylov subspace methods, in: 2013 European control
  conference (ECC), IEEE, 2013, pp. 3427--3432.

\bibitem{ibrir2018projection}
S.~Ibrir, A projection-based algorithm for model-order reduction with $\mathcal{H}_2$
  performance: A convex-optimization setting, Automatica 93 (2018) 510--519.

\bibitem{kundur1994power}
P.~Kundur, N.~J. Balu, M.~G. Lauby, Power system stability and control, Vol.~7,
  McGraw-hill New York, 1994.

\bibitem{pal2006robust}
B.~Pal, B.~Chaudhuri, Robust control in power systems, Springer Science \&
  Business Media, 2006.

\bibitem{rogers2012power}
G.~Rogers, Power system oscillations, Springer Science \& Business Media, 2012.

\bibitem{grimble1979solution}
M.~Grimble, Solution of finite-time optimal control problems with mixed end
  constraints in the s-domain, IEEE Transactions on Automatic Control 24~(1)
  (1979) 100--108.

\bibitem{gawronski1990model}
W.~Gawronski, J.-N. Juang, Model reduction in limited time and frequency
  intervals, International Journal of Systems Science 21~(2) (1990) 349--376.

\bibitem{redmann2018output}
M.~Redmann, P.~K{\"u}rschner, An output error bound for time-limited balanced
  truncation, Systems \& Control Letters 121 (2018) 1--6.

\bibitem{redmann2020lt2}
M.~Redmann, An $L_T^2$-error bound for time-limited balanced truncation,
  Systems \& Control Letters 136 (2020) 104620.

\bibitem{gugercin2004survey}
S.~Gugercin, A.~C. Antoulas, A survey of model reduction by balanced truncation
  and some new results, International Journal of Control 77~(8) (2004)
  748--766.

\bibitem{kurschner2018balanced}
P.~K{\"u}rschner, Balanced truncation model order reduction in limited time
  intervals for large systems, Advances in Computational Mathematics 44~(6)
  (2018) 1821--1844.

\bibitem{zulfiqar2020ptime}
U.~Zulfiqar, M.~Imran, A.~Ghafoor, M.~Liaqat, Time/frequency-limited
  positive-real truncated balanced realizations, IMA Journal of Mathematical
  Control and Information 37~(1) (2020) 64--81.

\bibitem{benner2021frequency}
P.~Benner, S.~W. Werner, Frequency-and time-limited balanced truncation for
  large-scale second-order systems, Linear Algebra and its Applications 623
  (2021) 68--103.

\bibitem{kumar2017generalized}
D.~Kumar, A.~Jazlan, V.~Sreeram, Generalized time limited gramian based model
  reduction, in: 2017 Australian and New Zealand Control Conference (ANZCC),
  IEEE, 2017, pp. 47--49.

\bibitem{imran2022development}
M.~Imran, M.~Imran, Development of time-limited model reduction technique for
  one-dimensional and two-dimensional systems with error bound via balanced
  structure, Journal of the Franklin Institute (2022).

\bibitem{shaker2014time}
H.~R. Shaker, M.~Tahavori, Time-interval model reduction of bilinear systems,
  International Journal of Control 87~(8) (2014) 1487--1495.

\bibitem{tahavori2013model}
M.~Tahavori, H.~R. Shaker, Model reduction via time-interval balanced
  stochastic truncation for linear time invariant systems, International
  Journal of Systems Science 44~(3) (2013) 493--501.

\bibitem{goyal2019time}
P.~Goyal, M.~Redmann, Time-limited $\mathcal{H}_2$-optimal model order
  reduction, Applied Mathematics and Computation 355 (2019) 184--197.

\bibitem{sinani2019h2}
K.~Sinani, S.~Gugercin, $\mathcal{H}_2(t_f)$ optimality conditions for a
  finite-time horizon, Automatica 110 (2019) 108604.

\bibitem{das2022h}
K.~Das, S.~Krishnaswamy, S.~Majhi, $\mathcal{H}_2$ optimal model order
  reduction over a finite time interval, IEEE Control Systems Letters 6 (2022)
  2467--2472.

\bibitem{zulfiqar2021frequency}
U.~Zulfiqar, V.~Sreeram, X.~Du, On frequency-and time-limited
  $\mathcal{H}_2$-optimal model order reduction, arXiv preprint
  arXiv:2102.03603.

\bibitem{zulfiqar2021h2}
U.~Zulfiqar, $\mathcal{H}_2$ near-optimal projection methods for finite-horizon
  model order reduction, Ph.D. thesis, The University of Western Australia,
  Perth, Australia (2021).

\bibitem{zulfiqar2020time}
U.~Zulfiqar, V.~Sreeram, X.~Du, Time-limited pseudo-optimal
  $\mathcal{H}_2$-model order reduction, IET Control Theory \& Applications
  14~(14) (2020) 1995--2007.

\bibitem{antoulas2005approximation}
A.~C. Antoulas, Approximation of large-scale dynamical systems, SIAM, 2005.

\bibitem{zhou1995frequency}
K.~Zhou, Frequency-weighted $\mathcal{H}_\infty$ norm and optimal Hankel norm
  model reduction, IEEE Transactions on Automatic Control 40~(10) (1995)
  1687--1699.

\bibitem{zhou1996robust}
K.~Zhou, J.~Doyle, K.~Glover, Robust and optimal control, Control Engineering
  Practice 4~(8) (1996) 1189--1190.

\bibitem{bini2016computing}
D.~A. Bini, S.~Dendievel, G.~Latouche, B.~Meini, Computing the exponential of
  large block-triangular block-toeplitz matrices encountered in fluid queues,
  Linear Algebra and its Applications 502 (2016) 387--419.

\bibitem{smalls2007exponential}
N.~N. Smalls, The exponential function of matrices, Master thesis, Georgia State University (2007).

\bibitem{petersen2008matrix}
K.~B. Petersen, M.~S. Pedersen, et~al., The matrix cookbook, Technical
  University of Denmark 7~(15) (2008) 510.

\bibitem{petersson2013nonlinear}
D.~Petersson, A nonlinear optimization approach to $\mathcal{H}_2$-optimal modeling and
  control, Ph.D. thesis, Link{\"o}ping University Electronic Press (2013).

\bibitem{benner2001efficient}
P.~Benner, E.~S. Quintana-Ort{\'\i}, G.~Quintana-Ort{\'\i}, Efficient numerical
  algorithms for balanced stochastic truncation, International Journal of Applied Mathematics and Computer Science 11~(5) (2001) 1123-1150.

\bibitem{anic2013interpolatory}
B.~Ani{\'c}, C.~Beattie, S.~Gugercin, A.~C. Antoulas, Interpolatory
  weighted-$\mathcal{H}_2$ model reduction, Automatica 49~(5) (2013)
  1275--1280.

\bibitem{breiten2015near}
T.~Breiten, C.~Beattie, S.~Gugercin, Near-optimal frequency-weighted
  interpolatory model reduction, Systems \& Control Letters 78 (2015) 8--18.

\bibitem{zulfiqar2022frequency}
U.~Zulfiqar, V.~Sreeram, M.~I. Ahmad, X.~Du, Frequency-weighted
  $\mathcal{H}_2$-optimal model order reduction via oblique projection,
  International Journal of Systems Science 53~(1) (2022) 182--198.

\bibitem{rommes2006efficient}
J.~Rommes, N.~Martins, Efficient computation of multivariable transfer function
  dominant poles using subspace acceleration, IEEE transactions on power
  systems 21~(4) (2006) 1471--1483.

\bibitem{martins2007computation}
N.~Martins, P.~C. Pellanda, J.~Rommes, Computation of transfer function
  dominant zeros with applications to oscillation damping control of large
  power systems, IEEE transactions on power systems 22~(4) (2007) 1657--1664.

\bibitem{davis2004algorithm}
T.~A. Davis, Algorithm 832: UMFPACK v4. 3---an unsymmetric-pattern multifrontal
  method, ACM Transactions on Mathematical Software (TOMS) 30~(2) (2004)
  196--199.

\bibitem{demmel1999supernodal}
J.~W. Demmel, S.~C. Eisenstat, J.~R. Gilbert, X.~S. Li, J.~W. Liu, A supernodal
  approach to sparse partial pivoting, SIAM Journal on Matrix Analysis and
  Applications 20~(3) (1999) 720--755.

\bibitem{panzer2014model}
H.~K. Panzer, Model order reduction by Krylov subspace methods with global
  error bounds and automatic choice of parameters, Ph.D. thesis, Technische
  Universit{\"a}t M{\"u}nchen (2014).

\bibitem{chahlaoui2005benchmark}
Y.~Chahlaoui, P.~V. Dooren, Benchmark examples for model reduction of linear
  time-invariant dynamical systems, in: Dimension reduction of large-scale
  systems, Springer, 2005, pp. 379--392.

\end{thebibliography}

\end{document}